\documentclass[a4paper,UKenglish,cleveref, autoref, thm-restate]{lipics-v2021}
\hideLIPIcs  %

\usepackage{listings}
\usepackage{xcolor} %
\usepackage{tcolorbox}
\usepackage{amsmath}
\usepackage{stmaryrd}

\usepackage{hyperref}

\usepackage{sidecap}
\usepackage{wrapfig}
\usepackage{amsmath}
\usepackage{amssymb}
\usepackage{floatflt}
\usepackage{comment}
\usepackage{todonotes}

\newcommand{\changed}[1]{\textcolor{blue}{#1}}

\usepackage{xspace} %
\usepackage{tikz}
\usetikzlibrary{positioning,arrows.meta,decorations.pathmorphing}
\usepackage{cleveref} %

\newcommand{\RPT}{\ensuremath{\mathsf{RPT}}\xspace}
\newcommand{\HSC}{\ensuremath{\mathsf{HSC}}\xspace}
\newcommand{\hsc}{\ensuremath{\mathsf{hsc}}\xspace}
\newcommand{\RDA}{\ensuremath{\mathsf{RDA}}\xspace}
\newcommand{\RDAs}{\ensuremath{\mathsf{RDAs}}\xspace}

\newcommand{\ARDA}{\ensuremath{\mathsf{ARDA}}\xspace}
\newcommand{\ARDAs}{\ensuremath{\mathsf{ARDAs}}\xspace}

\newcommand{\RDTS}{\ensuremath{\mathsf{RDTS}}\xspace}
\newcommand{\PRDTS}{\ensuremath{\mathsf{PRDTS}}\xspace}
\newcommand{\HSCL}{\ensuremath{\mathsf{HSCL}}\xspace}
\newcommand{\LTL}{\ensuremath{\mathsf{LTL}}\xspace}
\newcommand{\RMA}{\ensuremath{\mathsf{RMA}}\xspace}
\newcommand{\PIA}{\ensuremath{\mathsf{PIA}}\xspace}
\newcommand{\SSR}{\ensuremath{\mathsf{SSR}}\xspace}
\newcommand{\BWA}{\ensuremath{\mathsf{BWA}}\xspace}
\newcommand{\RIA}{\ensuremath{\mathsf{RIA}}\xspace}
\newcommand{\HR}{\ensuremath{\mathsf{HR}}\xspace}
\newcommand{\HRE}{\ensuremath{\mathsf{HRE}}\xspace}
\newcommand{\HRecord}{\ensuremath{\mathsf{HRecord}}\xspace}
\newcommand{\local}{\ensuremath{\mathsf{local}}\xspace}
\newcommand{\flocal}{\ensuremath{\mathsf{local}}\xspace}
\newcommand{\cumul}{\ensuremath{\mathsf{cumul}}\xspace}

\newcommand{\RDTSRMA}{\ensuremath{\mathsf{RDTS}\textsuperscript{$\mathsf{RMA}$}}\xspace}
\newcommand{\cfg}{\mathsf{cfg}\xspace}
\newcommand{\PState}{\mathsf{PState}\xspace}
\newcommand{\NState}{\mathsf{NState}\xspace}
\newcommand{\Update}{\mathsf{Update}\xspace}
\newcommand{\Recv}{\mathsf{Recv}\xspace}
\newcommand{\enabled}{\mathsf{enabled}\xspace}
\newcommand{\effect}{\mathsf{effect}\xspace}
\newcommand{\Rules}{\mathsf{Rules}\xspace}

\newcommand{\rcvMsg}{\mathsf{rcvMsg}\xspace}
\newcommand{\loc}{\mathsf{loc}\xspace}
\newcommand{\rd}{\mathsf{rd}\xspace}
\newcommand{\frm}{\mathsf{frm}\xspace}
\newcommand{\type}{\mathsf{type}\xspace}
\newcommand{\guard}{\mathsf{guard}\xspace}

\newcommand{\tgt}{\mathsf{tgt}\xspace}
\newcommand{\rmax}{r_\mathsf{max}\xspace}
\newcommand{\rmin}{r_\mathsf{min}\xspace}

\newcommand{\LinTrm}{\mathsf{LinTrm}\xspace}

\usepackage{booktabs}
\usepackage[T1]{fontenc}
\usepackage{graphicx}

\title{Parameterized Verification of Asynchronous Round-Based Distributed Algorithms via Reduction to Finite-Counter Systems} %

\titlerunning{Parameterized Verification of Round-Based Distributed Algorithms} %

\author{Nathalie Bertrand}{Univ Rennes, Inria, CNRS, IRISA, France}{}{}{}
\author{Pranav Ghorpade}{The University of Sydney, Australia}{}{}{}
\author{Sasha Rubin}{The University of Sydney, Australia}{}{}{}
\Copyright{Nathalie Bertrand, Pranav Ghorpade and Sasha Rubin} %
\authorrunning{N. Bertrand, P. Ghorpade and S. Rubin}


\category{} %

\nolinenumbers %

\EventEditors{~}
\EventNoEds{1}
\EventLongTitle{}%
\EventShortTitle{}%
\EventAcronym{}%
\EventYear{2026}
\EventDate{September 2--5, 2026}
\EventLocation{Liverpool, United Kingdom}
\EventLogo{}
\SeriesVolume{42}
\ArticleNo{1}

\begin{document}

\maketitle

\begin{abstract}
Traditional model-checking techniques typically verify distributed algorithms only for a fixed number of finite-state processes. Parameterized model checking generalizes this to any number of processes, while still typically assuming that each process is finite-state. In this work, we consider asynchronous round-based distributed algorithms in which each process is infinite-state since it can execute for an infinite number of rounds. %
We show that the parameterized verification problem for asynchronous round-based distributed algorithms is undecidable, already for simple specifications. Nevertheless, as our main contribution, we provide a reduction to \(\LTL\) model checking over finite-counter systems and prove that it is sound and complete. This enables the use of off-the-shelf, mature symbolic model checkers for finite-counter systems. We demonstrate the practical applicability of this reduction by verifying safety and liveness properties of several asynchronous round-based consensus and leader-election algorithms using the nuXmv model checker.

\end{abstract}

\section{Introduction}
\label{sec:intro}

Asynchronous distributed algorithms consist of a parameterized number of identical processes that asynchronously exchange messages with other processes and update their local states.
A subclass of these is that of asynchronous round-based distributed algorithms ($\ARDAs$), in which each process progresses through an unbounded sequence of rounds and each message is tagged with the sender's current round number. Typical examples include Ben-Or's randomized consensus algorithms~\cite{ben-or}, Bracha's consensus algorithm~\cite{bracha}, and Raft's leader-election algorithm~\cite{raft}. In the absence of synchronization, processes in such algorithms may drift across rounds, so each process must store its current round index as part of its local state. Since this index is unbounded, each process has infinitely many possible states. Consequently, the global system exhibits two orthogonal sources of unboundedness: the parameterized number of processes and the unbounded round index of each process.  

Parameterized verification of distributed algorithms refers to checking their correctness for 
all numbers $n$ of processes, provided that a resilience condition relating $n$ and the maximum number $t$ of faulty processes is satisfied (e.g., $3t < n$). 
Parameterized model checking provides techniques that cope with all admissible valuations of the parameters $n$ and $t$, namely those satisfying the resilience condition. However, it typically assumes that processes are finite-state (see, for instance, the book~\cite{PMCP_book} or the recent survey~\cite{Konnov_2023}) and is therefore not directly applicable to asynchronous round-based distributed algorithms. Moreover, even for fixed parameter values, the resulting system is infinite-state because of the round indices and cannot be directly handled by finite-state model checkers such as SPIN, NuSMV, or TLA+. 

In this work, we introduce a modeling formalism for a broad class of $\ARDAs$, together with a specification logic designed to express their standard correctness properties. For this formalism, we show that the parameterized verification problem is undecidable, whereas the fixed-instance verification problem is decidable. To address the parameterized case despite this undecidability, we give a sound and complete reduction to $\LTL$ model checking over finite-counter systems. Although this target problem is itself undecidable, it is well studied, and sound, necessarily incomplete, symbolic model-checking techniques for it have been developed and implemented, for example in nuXmv~\cite{NuXmv}. Thus, our reduction enables the use of existing symbolic model checkers for parameterized verification of $\ARDAs$, while also potentially benefiting from future advances in such tools.
To assess the practical applicability of the reduction, we apply it to four case studies: (i) Ben-Or's consensus with crash faults~\cite[Algo A]{ben-or}, (ii) Ben-Or's consensus with Byzantine faults~\cite[Algo B]{ben-or}, (iii) Bracha's consensus with Byzantine faults~\cite[Fig. 4]{bracha}, and (iv) Raft's leader-election~\cite[Sec. 5.2]{raft}. Using nuXmv as the back end, we verify safety and liveness properties, namely agreement, validity, termination (under some condition), and leader uniqueness.
These case studies suggest that existing verification engines can efficiently handle the resulting finite-counter systems.

\medskip
\noindent \textbf{Brief Overview.}
\S~\ref{subsec:model} introduces templates for $\ARDAs$. 
A template specifies the symbolic parameters of the algorithm together with the behavior of a single process. For each admissible parameter valuation, it induces a round-based distributed transition system ($\RDTS$) capturing the behavior of the resulting interacting processes. The parameterized semantics (parameterized $\RDTS$) is obtained as the disjoint union of these fixed-instance $\RDTS$ over all admissible parameter valuations.
\S~\ref{subsec:HSCL} identifies a class of \emph{action-based} temporal properties, called \emph{history state-count properties}, and introduces a tailored action-based logic, \emph{history state-count logic} ($\HSCL$). 
This logic is deliberately restricted so that it can be translated into state-based $\LTL$, while remaining expressive enough to capture standard correctness conditions for consensus and leader election. 
\S~\ref{subsec:Verif-Problems} gives a reduction from the non-halting problem for two-counter machines to the verification problem of parameterized $\RDTS$ against $\HSCL$ formulas, thereby establishing undecidability. \S~\ref{subsec:CS} gives four sound and complete transformations that reduce the parameterized $\RDTS$ to a finite-counter system while preserving history state-count properties. \S~\ref{subsec:LTL} then gives two further sound and complete transformations that reduce verification of parameterized $\RDTS$ against $\HSCL$ formulas to $\LTL$ model checking over finite-counter systems.
When specialized to fixed-instance semantics $\RDTS$, our reduction yields a \emph{one}-counter system for history state-count properties and reduces verification of $\RDTS$ against $\HSCL$ formulas to $\LTL$ model checking over \emph{finite-state} systems. This proves that the fixed-instance verification problem is decidable. \S~\ref{sec:evaluation} presents the case studies, and \S~\ref{sec:conclusion} concludes the paper.

\medskip
\noindent {\bf Related Work.} 
Round-based distributed algorithms belong to the broader class of fault-tolerant distributed algorithms, which also includes algorithms in which processes are finite-state, such as broadcast algorithms. 
Threshold automata~\cite{TA-1,POPL-2017} provide a successful framework for parameterized verification of such finite-state algorithms by abstracting them to finite-counter systems that can be analyzed by symbolic model checkers. 
This line of work was later extended to synchronous round-based distributed algorithms through synchronous threshold automata~\cite{DBLP:conf/tacas/StoilkovskaKWZ19}. 
For asynchronous round-based distributed algorithms ($\ARDAs$), a compositional approach was proposed in~\cite{round-rigid}. It reduces the verification to proof obligations for individual rounds, which are handled within the threshold-automata framework. A limitation is that the target property is not established directly: one must first identify suitable round invariants and then argue {separately} that these invariants imply the desired property. 

Layered threshold automata~\cite{bertrand_2021}, together with the PyLTA prototype~\cite{PyLTA}, were later proposed as a more direct framework for verifying $\ARDAs$. 
However, this framework is restricted to algorithms that admit a \emph{layered} structure, thereby excluding, for instance, most leader-election algorithms.
More importantly, the corresponding verification procedure does not retain the main tooling advantage of standard threshold automata. It first gives a sound and complete abstraction to counter systems with \emph{infinitely} many counters, and then applies a further abstraction to finite-state systems that is \emph{incomplete}. As a result, the approach relies on a bespoke CEGAR loop rather than on off-the-shelf symbolic model checking.
In contrast, our work targets a broader class of $\ARDAs$, including leader-election algorithms, and gives a sound and complete reduction to finite-counter systems. This enables the use of mature general-purpose symbolic model checkers, as illustrated in Fig.~\ref{fig:contribution}.

\begin{figure}[t]
  \centering
  \scalebox{.9}{\begin{tikzpicture}
\draw[teal!30!white, fill=teal!30!white, very thick] (0,0) rectangle (2,3);
\draw[orange!60!white,fill=orange!60!white,very thick] (1,0.75) ellipse (.75 and 0.5);
\node (lay) at (1,.9) {{\scriptsize Layered}};
\node (lay) at (1,.6) {{\scriptsize $\ARDAs$}};
\node (lay) at (1,2) {{\scriptsize $\ARDAs$}};

\draw[teal!30!white, fill=teal!30!white, very thick] (5,1.5) rectangle (7,3);
\node (lay) at (6,2.4) {{\scriptsize Finite-counter}};
\node (lay) at (6,2.1) {{\scriptsize system}};

\draw[orange!60!white, fill=orange!60!white, very thick] (5,.25) rectangle (7,1.25);
\node (lay) at (6,.9) {{\scriptsize Infinite-counter}};
\node (lay) at (6,.6) {{\scriptsize system}};

\draw[-stealth,decorate,decoration={snake, segment length=4.15mm, amplitude=.8mm},color=teal,very thick] (2.05,2.25) -- (4.95,2.25) node[above,midway,color=black] {{\scriptsize sound and complete}}
node[below,midway,color=black] {{\scriptsize reduction}};
\draw[-stealth,decorate,decoration={snake, segment length=4.15mm, amplitude=.8mm},color=orange,very thick] (1.8,.75) -- (4.95,.75)  node[above,midway,color=black] {{\scriptsize sound and complete}}node[below,midway,color=black] {{\scriptsize reduction~\cite{bertrand_2021}}};

\draw[gray!30!white, fill=gray!30!white, very thick] (8,1.5) rectangle (11,3);
\node (lay) at (9.5,2.7) {{\scriptsize Off-the-shelf}};
\node (lay) at (9.5,2.4) {{\scriptsize verification engines:}};
\node (lay) at (9.5,2.1) {{\scriptsize IC3, bounded MC}};
\node (lay) at (9.5,1.8) {{\scriptsize interpolation, $\ldots$}};
\draw[-latex,color=teal,very thick] (7.05,2.25) -- (7.95,2.25);

\colorlet{grorange}{orange!80!gray}
\draw[grorange!40!white,fill=grorange!40!white, very thick] (8,.25) rectangle (10,1.25);
\node (lay) at (9,.9) {{\scriptsize Customized}};
\node (lay) at (9,.6) {{\scriptsize CEGAR~\cite{PyLTA}}};
\draw[-latex,color=orange,very thick] (7.05,.75) -- (7.95,.75);
\end{tikzpicture}}
  \caption{Our reduction improves over the verification framework for layered $\ARDAs$ by~\cite{bertrand_2021,PyLTA} in both scope and the ability to reuse mature verification tools.
  }
  \label{fig:contribution}
\end{figure}
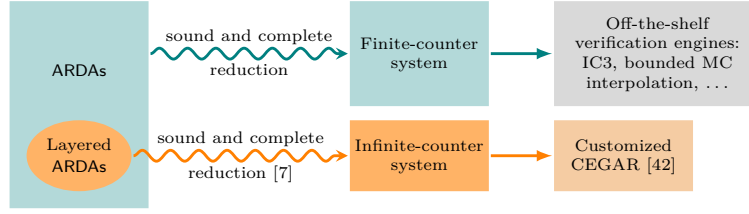

Traditionally, parameterized verification of $\ARDAs$ has largely relied on human-guided methods.
Interactive theorem provers such as Rocq, Isabelle/HOL, and TLA+ have been used to formalize distributed algorithms and establish their correctness through mechanized, expert-driven proofs~\cite{raft-coq,paxos-tla,DAGconsensus-tla}.
Deductive verification frameworks such as Ivy~\cite{Ivy} encode distributed algorithms in first-order logic and can prove correctness automatically once suitable inductive invariants are provided. However, discovering invariants remains challenging, motivating a large body of work on automatic inductive-invariant inference~\cite{ind-inf-2,ind-inf-3,inv-inf,ind-inv-4}. 

For fixed-instance verification, prior work~\cite{DTT-gandalf14,NTK-prdc12} applies finite-state model-checking tools to prove correctness for a fixed number of processes executing a bounded number of rounds. 
In contrast, our setting imposes no bound on the number of rounds. 
The work of~\cite{reduction-thm} also addresses fixed-instance verification of round-based consensus algorithms in the Heard-Of model~\cite{HO-Model}, using TLC to verify fixed instances of algorithms such as OneThirdRule and UniformVoting. 
However, their reduction is tailored to TLA+/TLC, whose specifications are written in the Temporal Logic of Actions, whereas our reduction targets standard $\LTL$ model checking. 
Moreover, our framework handles a broader class of algorithms, including Byzantine fault-tolerant consensus algorithms and leader-election algorithms.

\section{Modeling Asynchronous Round-Based Distributed Algorithms}
\label{sec:modelling}
This section introduces the process templates that serve as our formal model, together with the specification language and verification problems studied in the paper.

\subsection{Round-based Process Templates}
\label{subsec:model}
The syntax of our model is motivated by the typical structure of $\ARDAs$; see Fig.~\ref{alg:voting} (left) for an example. These algorithms are parameterized; processes communicate via broadcast, with choices governed by threshold guards; and the overall behavior of each process is composed of identical, bounded, communication-closed rounds. 
We now expand on these three aspects. 
(i) $\ARDAs$ are parameterized by integer variables (typically $n$ for number of processes and $t$ for an upper-bound on the number of faulty processes) and are designed to work when these variables satisfy a \emph{resilience condition}, typically expressed as a linear arithmetic formula (e.g., $t < n/2$). 
(ii) In $\ARDAs$, communication is typically by \emph{broadcast}. Moreover, since processes may be faulty, an $\ARDA$ cannot rely on receiving messages from specific processes. 
Accordingly, process updates are governed by \emph{threshold guards}, which depend only on the number of received messages of each type, rather than on the identities of the senders; see step 3 in Fig.~\ref{alg:voting} (left).
(iii) In $\ARDAs$, the behavior of a process in each round follows the same high-level structure: a finite sequence of broadcast, receive, and compute steps, potentially followed by a {progress step} to a future round. As a result, rounds are \emph{structurally identical}, and the control-flow within each round is \emph{bounded}. Messages from different rounds are distinguished using \emph{round tags},
which enforce \emph{communication closure} within rounds: messages tagged with round $r$ are used exclusively during computations in round $r$.

Exploiting the identical structure of rounds in $\ARDAs$, our round-based process template provides a round-independent representation capturing the behavior of a \emph{single process} within a \emph{single round}, together with the rules governing round progress. Similar representations for $\ARDAs$ have appeared in prior work, though in more restrictive forms \cite{bertrand_2021} or in a less direct form requiring additional modeling components (e.g., auxiliary locations)~\cite{round-rigid}.

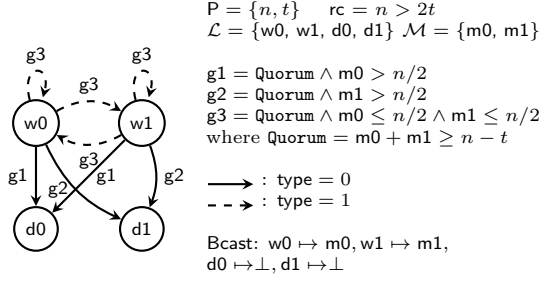
\begin{figure}[t]
\centering
\begin{minipage}{0.38\textwidth}
\scriptsize
Parameters: $n, t$ , with $n > 2t$\\
Input: $v_0 \in \{0, 1\}$

Initialize: $v := v_0$, $d := \texttt{Nil}$, $r := 0$

At round $r$:
\begin{enumerate}
    \item \texttt{Broadcast} ($v$, $r$)
    \item \texttt{Wait} until $\geq (n - t)$ messages\\
    received with round tag $r$
    \item \texttt{Compute} if some value $v'$ appears \\
        in $> n/2$ messages: $d := v'$; \texttt{halt},\\
        else $v := 0$ or $v := 1$
    \item \texttt{Progress} $r := r + 1$
\end{enumerate}
 
\end{minipage}\hfill
\begin{minipage}{0.23\textwidth}
\begin{flushright}
\scriptsize
\begin{tikzpicture}[->, >=stealth, node distance=1.4cm, thick, state/.style={circle, draw, minimum size=0.4cm}]
  \node[state] (w0) {$\mathsf{w0}$};
  \node[state, right of=w0] (w1) {$\mathsf{w1}$};
  \node[state, below of=w0] (d0) {$\mathsf{d0}$};
  \node[state, below of=w1] (d1) {$\mathsf{d1}$};
   \path (w0) edge node[left] {$\mathsf{g1}$} (d0);
  \path (w1) edge node[right] {$\mathsf{g1}$} (d0);
  \path (w0) edge[bend right=20] node[left] {$\mathsf{g2}$} (d1);
  \path (w1) edge[bend left=20] node[right] {$\mathsf{g2}$} (d1);
  \path (w0) edge[dashed, bend left] node[above] {$\mathsf{g3}$} (w1);
  \path (w0) edge[dashed,loop above] node[above] {$\mathsf{g3}$} (w0);
  \path (w1) edge[dashed,loop above] node[above] {$\mathsf{g3}$} (w1);
  \path (w1) edge[dashed,bend left] node[below] {$\mathsf{g3}$} (w0);
\end{tikzpicture}
\end{flushright}
\end{minipage}\hfill
\begin{minipage}{0.36\textwidth}
\scriptsize
$\textsf{P}=\{ n, t\}$ \quad \textsf{rc} = $n > 2t$\\
\quad $\mathcal{L} = \{\textsf{w0, w1, d0, d1}\}$
$\mathcal{M} = \{\textsf{m0, m1}\}$\\

$\mathsf{g1} = \mathtt{Quorum} \wedge \textsf{m0} > n/2$\\
$\mathsf{g2} =\mathtt{Quorum} \wedge \textsf{m1} > n/2$\\
$\mathsf{g3} = \mathtt{Quorum} \wedge \textsf{m0} \leq n/2 \wedge \textsf{m1} \leq n/2$\\
where $\mathtt{Quorum}= \textsf{m0} + \textsf{m1} \geq n - t$\\

\tikz \draw[->, >=stealth, thick] (0,0) -- (0.6,0); : $\textsf{type} =0$\\
\tikz \draw[dashed,->,>=stealth, thick] (0,0) -- (0.6,0); : $\textsf{type} = 1$\\

\textsf{Bcast}: $\textsf{w0}\mapsto\textsf{m0}, \textsf{w1}\mapsto \textsf{m1},$\\ $\textsf{d0}\mapsto \perp, \textsf{d1}\mapsto\perp$\\
\end{minipage}
\caption{{\bf Left}: Pseudo-code of a simplified round-based voting algorithm.
{\bf Right}: Corresponding process template representation. 
In the guards, the variable \textsf{mi} denotes the number of messages of type \textsf{mi} received by the process. }
\label{fig:RPT-example}
 \label{alg:voting}
\end{figure}

Before proceeding to the formal definition of templates, we illustrate it on the algorithm in Fig.~\ref{alg:voting} (left) for which Fig.~\ref{fig:RPT-example} (right) gives a process template.
The set of parameters is $\mathsf{P} = \{n, t\}$, with admissible valuations satisfying the resilience condition $\mathsf{rc} = t <n/2$. 
The locations $\mathcal{L}= \{\mathsf{w0, w1, d0, d1}\}$ represent a process's position within a round: either in a \emph{wait location} ($\mathsf{w0}$ or $\mathsf{w1}$, depending on its current vote) or in a \emph{decide location} ($\mathsf{d0}$ or $\mathsf{d1}$, depending on its decision).
These locations correspond to the \emph{wait} and \emph{compute} steps in the pseudo-code. 
The \emph{broadcast} and \emph{progress} steps are treated implicitly: after updating its vote, a process immediately progresses to the next round, broadcasts its new vote, and re-enters a wait location. 
The message types $\mathcal{M} = \{\mathsf{m0}, \mathsf{m1}\}$ represent the possible votes a process may broadcast. The transition rules, represented by edges, are labeled with a threshold guard on received messages. For example, guard \textsf{g1} corresponds to the conjunction ``at least $n-t$ messages are received'' and ``more than $n/2$ are of type $\mathsf{m0}$''. Transitions are also equipped with a \textsf{type} that specifies how the round id evolves: type $0$ indicates that the process remains in its current round, while in general type $k$ indicates that it progresses by $k$ rounds. Our example only has type $0$ and type $1$ transitions represented by solid and dashed edges respectively. More generally, round-tag increases of up to $b \in \mathbb{N}$ are allowed, where $b$ is referred to as the jump bound.
A broadcast function \textsf{Bcast} maps each location to the message sent while entering it: here, for $i = 1,2$, \textsf{wi} maps to \textsf{mi}, and \textsf{di} maps to $\perp$ (i.e., no broadcast). 
As exhibited by this example, we assume that the subgraph formed by the solid edges (i.e., those with $\textsf{type} = 0$) is acyclic. 

\begin{definition}[Round-based process templates]\label{def:RPT}
Let $b \in \mathbb{N}$ be a \emph{jump bound}. A \emph{round-based process template}, or simply a \emph{template}, is a tuple
$\langle \mathsf{P}, \mathsf{rc}, \mathcal{L}, \mathcal{I}, \mathcal{M}, \mathsf{Bcast}, \mathsf{Rules} \rangle$
where:
\begin{itemize}
    \item $\mathsf{P}$ is a finite set of \emph{symbolic parameters}, including $n$, the number of processes;
    \item $\mathsf{rc} \in \mathsf{LinArith}(\mathsf{P})$ is a \emph{resilience condition};
    \item $\mathcal{L}$ is a finite set of \emph{locations}, and $\mathcal{I} \subseteq \mathcal{L}$ is the set of \emph{initial locations};
    \item $\mathcal{M}$ is a finite set of \emph{message types};
    \item $\mathsf{Bcast} : \mathcal{L} \to \mathcal{M} \cup \{\perp\}$ is the \emph{broadcast function};
    \item $\mathsf{Rules}$ is a finite set of \emph{transition rules}, where transition rule $\rho$ is a record of the form %
    $\rho = \{\mathsf{from} \in \mathcal{L},\ \mathsf{to} \in \mathcal{L},\ \mathsf{type} \in \llbracket 0,b \rrbracket,\ \mathsf{guard} \in \mathsf{LinArith}(\mathcal{M} \cup \mathsf{P})\}.$
    The set $\mathsf{Rules}$ is required to satisfy the following conditions: 
    \begin{enumerate}
        \item the subgraph induced by the rules $\rho$ with $\rho.\mathsf{type}=0$ is acyclic;
        \item no rule $\rho$ with $\rho.\mathsf{type}=0$ has $\rho.\mathsf{to} \in \mathcal{I}$;
        \item for all rules $\rho,\rho' \in \mathsf{Rules}$ such that $\rho.\mathsf{type}=0$ and $\rho'.\mathsf{from}$ is reachable from $\rho.\mathsf{to}$ by a finite sequence of type-$0$ rules, the guard $\rho'.\mathsf{guard}$ is monotone in the message variables occurring in $\rho.\mathsf{guard}$.
    \end{enumerate}
\end{itemize}
\end{definition}

In this definition, $\mathsf{rc}$ specifies the admissible parameter valuations; $\mathcal{L}$ and $\mathcal{M}$ denote, respectively, the locations a process may occupy and the message types it may send in a round; and each rule specifies a source location, a target location, a round increment of at most $b$, and the guard under which it is enabled. Condition~$1$ ensures that each process visits each location at most once within a round. Condition~$2$ ensures that processes enter initial locations only at the start of a new round. Condition~$3$ is a technical monotonicity condition adapted from~\cite{RMA-TA-ATVA20}. It is needed for the completeness proof of the received-message abstraction, one of the six transformations presented in \S~\ref{sec:abstractions}. Importantly, the overall reduction remains sound even without this condition. 
Moreover, as argued in~\cite{RMA-TA-ATVA20}, the monotonicity condition is often implicitly assumed  by distributed algorithms designers, and are thus satisfied by a wide range of distributed algorithms, such as the ones from~\cite{bracha,ben-or,song_bosco_2008,10.1145/226643.226647,DBLP:journals/dc/SrikanthT87}. %

Informally, Condition~$3$ ensures that once a process takes a type-$0$ rule, increasing the message counts that appear in that rule's guard cannot disable any rule that may be taken later in the same round. For example, if $\rho.\mathsf{guard}=\mathsf{m_1}+\mathsf{m_2}<\mathsf{n}/2$, then every rule $\rho'$ whose source location is reachable from $\rho.\mathsf{to}$ by a finite sequence of type-$0$ rules, and hence may be taken after $\rho$ within the same round, must have a guard that is monotone in $\mathsf{m_1}$ and $\mathsf{m_2}$. Here, a guard $g \in \mathsf{LinArith}(X)$ is \emph{monotone} in a set $Y \subseteq X$ if, for all valuations $\mu,\mu' : X \to \mathbb{N}$, whenever $\mu$ satisfies $g$, $\mu'(y)\geq \mu(y)$ for all $y\in Y$, and $\mu'(x)=\mu(x)$ for all $x\in X\setminus Y$, it also holds that $\mu'$ satisfies $g$.

\bigskip
While a template describes the symbolic parameters and the behavior of a single process, the semantics of several interacting processes is represented by an infinite-state transition system. We first define the fixed-instance semantics, for a given parameter valuation, and then lift it to the parameterized semantics over all admissible parameter valuations.

Given a template $\mathcal{T}$ and a parameter valuation $\nu : \mathsf{P} \to \mathbb{N}$ such that $\nu \models \mathsf{rc}$, 
the system of $\nu(n)$ identical processes executing $\mathcal{T}$ is represented by the action-labeled transition system $\RDTS(\mathcal{T},\nu) = \langle S_\nu, I_\nu, Act_\nu, Tr_\nu\rangle$. 
Its semantics is sketched here and detailed in App.~\ref{app:semantics-fixed}. Elements of $S_\nu$ are \emph{global states} (or configurations), each composed of a process state and a network state. 
The \emph{process state} $\PState :\llbracket 1, \nu(n) \rrbracket \rightarrow \langle \mathsf{loc}:\mathcal{L},\mathsf{rd}:\mathbb{N},\mathsf{rcvMsg}:\mathcal{M} \times \mathbb{N} \rightarrow \mathbb{N} \rangle$ stores for each process: its current round, location, and multiset of received message types. 
The latter records how many messages of a given type were received with a given round tag. The \emph{network state}  $\NState: \mathcal{M} \times \mathbb{N} \rightarrow \mathbb{N}$ collects the number of messages of each type and round tag that were broadcast. 
In every configuration, the number of received messages (in process state) is pointwise bounded by the number of broadcast messages (in network state). Initial global states $I_\nu$ are those where all processes are in initial locations at round 0 with empty set of received messages, and the network state contains no
message. 
Actions in $Act_\nu$ are of two types: (i) a reception $\textsf{Receive}(i,\langle m,r\rangle)$, where process $i$ receives a message of type $m$ tagged with round $r$, or (ii) an update $\Update(i,\rho,r)$, where process $i$ updates its state according to rule $\rho$ of the template at round $r$. A message reception is enabled for a process if its count of received messages (of that type and round tag) is strictly less than the number of similar broadcast messages in the network state; the process then increments its message count accordingly.  An update $\Update(i,\rho,r)$ is enabled if process $i$ is at location $\rho.\frm$ in round $r$, and its received messages with round tag $r$ satisfy the guard $\rho.\guard$; the process location and its round are then updated to $\rho.\textsf{to}$ and $r + \rho.\type$, and the process broadcasts a message $\textsf{Bcast}(\rho.\textsf{to})$ {(unless $\rho.\textsf{to} = \perp$)} tagged with round $r + \rho.\type$ {on the network state}.  An \emph{execution} of $\RDTS(\mathcal{T},\nu)$ is a finite or infinite alternating sequence of configurations and transitions starting from an initial configuration. 

To reason about all admissible parameter valuations simultaneously, we lift the fixed-instance transition system $\RDTS$ to its parameterized counterpart \(\PRDTS(\mathcal{T})\), which takes the disjoint union of all $\RDTS(\mathcal{T},\nu)$ for $\nu \models \mathsf{rc}$ (see App.~\ref{app:param-sem}).

\subsection{History State-Count Properties and Logic}\label{subsec:HSCL}
The choice of a property class is central to any verification framework, as it determines both the expressiveness of the specifications and the feasibility of the reductions. 
We therefore need a specification formalism that balances three requirements: 
(i) expressiveness, to capture standard correctness conditions for $\ARDAs$, as enumerated in Fig.~\ref{fig:prop-expressivity} (left); (ii) reduction-compatibility, to establish a sound and complete reduction from parameterized $\RDTS$ to finite-counter systems; (iii) $\LTL$-translatability, so that the final verification problem can be expressed as standard $\LTL$ model checking.

Previously introduced specification formalisms for threshold automata fall short in expressing standard correctness properties of $\ARDAs$. Specifically, \textsf{ELTL\textsubscript{FT}}~\cite{POPL-2017} lacks quantification over rounds and thus cannot express standard properties such as agreement.
The extension with universal round quantifiers~\cite{round-rigid} enables reasoning about properties local to individual rounds, but it still cannot express consensus properties that inherently span multiple rounds, such as agreement.
Finally,~\cite{bertrand_2021} does not define a formal specification logic and relies on assumptions about its successive abstractions that guarantee preservation of correctness.

{These limitations motivate our choice of specification formalism. We propose a semantic class of \emph{history state-count} ($\HSC$) properties, chosen to satisfy expressiveness and reduction-compatibility. We also introduce {$\HSC$ logic} ($\HSCL$), a syntactic fragment of temporal logic whose formulas denote a subset of $\HSC$ properties and are designed to be $\LTL$-translatable while retaining expressiveness. In contrast to prior state-based approaches~\cite{POPL-2017,bertrand_2021}, the $\HSC$ properties are action-based. This distinction is important for our reduction: although relevant action histories could, in principle, be encoded by auxiliary state predicates, the specifications and reduction are most naturally formulated in terms of the occurrence of transitions. Thus, instead of taking a restricted fragment of state-based $\LTL$ as the specification language, we use action-based $\HSCL$ and later translate it into $\LTL$ over the reduced finite-counter system.}

\subsubsection{History State-Count Properties}\label{subsubsec:HSC}
A History State-Count ($\HSC$) property applies to history state-count traces of executions.  For an execution $\pi$, its \emph{history state-count trace} (\hsc-trace) is a function $\hsc(\pi): \mathcal{L} \times \mathbb{N} \rightarrow \mathbb{N}$ that records, along $\pi$, the aggregate number of process visits to each location in every round. Formally, for a location $\ell$ and a round tag $r$, $\hsc(\pi)(\ell, r)$ denotes the number of update actions, with some rule $\rho$ and round tag $r'$ taken along $\pi$ such that $\rho.\textsf{to} = \ell$ and $r' + \rho.\textsf{type} = r$. Hence, the $\hsc$-trace is entirely determined by the multiset of pairs $(\rho, r)$ corresponding to the update actions occurring along the execution. Since each process can visit any location at most once per round (Condition~$1$, Def.~\ref{def:RPT}), the $\hsc$-trace equivalently captures, for every pair of a location and a round tag, the number of processes that have visited it.

For a template $\mathcal{T}$ with locations $\mathcal{L}$, a \emph{history state-count} ($\HSC$) property is a family \( \mathcal{P} = \{ \mathcal{P}_\nu \subseteq \mathsf{Dom}\ \mid\ \nu : \mathsf{P} \rightarrow \mathbb{N}\} \text{, where }\mathsf{Dom} = \{\, f : \mathcal{L}\times \mathbb{N} \to \mathbb{N} \,\}.\){ We write $\textsf{Traces}(\RDTS(\mathcal{T}, \nu)) = \{\textsf{hsc}(\pi)\ |\ \pi \text{ is an execution of $\RDTS(\mathcal{T}, \nu)$}\}$.} We write $\RDTS(\mathcal{T},\nu) \models \mathcal{P}$ iff $\textsf{Traces}(\RDTS(\mathcal{T}, \nu)) \subseteq \mathcal{P}_\nu$, and define $\PRDTS(\mathcal{T}) \models \mathcal{P}$ iff, for all valid parameter valuations $\nu$, $\RDTS(\mathcal{T},\nu) \models \mathcal{P}$.
\label{def:models-P}

\subsubsection{History State-Count Logic (\textsf{HSCL})}
Informally, \(\HSCL\) has two types of atomic constraints: (i) \emph{universal, round-local constraints} of the form $\forall r.\alpha_r$, which, for each round $r$, allow one to bound a weighted sum of visits to locations in that round by a threshold; and (ii) \emph{cumulative constraints}, written $\beta$, which allow one to bound the corresponding weighted sum aggregated over all rounds. Formulas in \(\HSCL\) are Boolean combinations (using $\neg,\land$) of these two types. Note that existential quantification over rounds is expressible as $\neg \forall r.\alpha_r \equiv \exists r.\neg \alpha_r$.

 \begin{definition}[$\HSC$]
   \label{def:HSCL}
   Let $\mathcal{X}$ be a finite set of locations and $\mathcal{Y}$ a
   finite set of parameters. Formulas of \emph{$\HSCL$} follow the grammar:
 \[\varphi \ ::= \ \forall r.\,{\alpha_r} \ \mid\ {\beta} \ \mid\ \neg \varphi \ \mid\ (\varphi \land \varphi)\]
 \[ \textrm{with }{\alpha_r} \ ::= \ \sum_{\ell \in \mathcal{X}} c_\ell \cdot \kappa(\ell,r)\ \le\ t, \quad {\beta} \ ::= \ \sum_{\ell \in \mathcal{X}} c_\ell \cdot \sum_{r \in \mathbb{N}} \kappa(\ell,r)\ \le\ t\]
 where $c_\ell \in \mathbb{N}$ are \emph{non-negative weights}; $\kappa(\ell,r)$ is a \emph{variable} (for $\ell \in \mathcal{X}$, $r \in \mathbb{N}$) interpreted as the number of process-visits to $\ell$ in round $r$; and the \emph{threshold} $t$ is a linear term over $\mathcal{Y}$ of the form $t = b_0 + \sum_{y \in \mathcal{Y}} b_y \cdot y$, with $b_0, b_y \in \mathbb{Z}$.
 We call instances of ${\alpha_r}$ \emph{round-local atoms} and instances of ${\beta}$ \emph{cumulative atoms}\footnote{Although the sum in $\beta$ may diverge, since $c_\ell \geq 0$ and $t$ is finite, satisfaction of $\beta$ is determined by a finite partial sum.}. 
 \end{definition} 

\label{sec:HSCLsemantics}
For a process template \(\mathcal{T}\) with a location set $\mathcal{L}$ and a parameter set $\mathsf{P}$, we interpret $\mathcal{X}$ as a subset of $\mathcal{L}$ and $\mathcal{Y}$ as a subset of $\mathsf{P}$. A $\HSCL$ formula $\varphi$ represents the $\HSC$ property $\mathcal{P}^\varphi = \{\mathcal{P}^\varphi_\nu \ \mid\ \nu :  \mathsf{P} \rightarrow \mathbb{N}\}$ where for a parameter valuation $\nu$, set $\mathcal{P}^\varphi_\nu= \{f : \mathcal{L}\times \mathbb{N} \to \mathbb{N}\ \mid\ f,\nu \models \varphi\}.$ Here, $f,\nu \models \varphi$ iff $\varphi$ evaluates to true under the substitution of each threshold term $t$ in $\varphi$ with its value under $\nu$ (i.e., $t [\mathcal{Y} \leftarrow \nu]$), and of each variable $\kappa(x,r)$ with $f(x,r)$.  We write $\RDTS(\mathcal{T},\nu)\models \varphi$ if $\RDTS(\mathcal{T},\nu)\models \mathcal{P}^\varphi$,and $\PRDTS(\mathcal{T})\models\varphi$ if $\PRDTS(\mathcal{T})\models \mathcal{P}^\varphi$.
\begin{figure}[t]
\begin{center}
\scriptsize 
\resizebox{\linewidth}{!}{
\begin{tabular}{l l}
      \multicolumn{2}{l}{\textbf{Agreement}: Either no correct process decides on 0, or no correct process decides on 1.} \\
      $\ \mathtt{A} := \forall r. \kappa(\textsf{d0}, r) \leq 0\ \vee \forall r. \kappa(\textsf{d1}, r) \leq 0$ \quad~& $\LTL(\mathtt{A})\ = \textbf{G}\bigl(\flocal(\textsf{d0}) \leq 0\bigr)\ \vee\ 
  \textbf{G}\bigl(\flocal(\textsf{d1}) \leq 0\bigr)$ \\[2mm]
     \multicolumn{2}{l}{\textbf{Validity}: If no correct process starts with 0, then no correct process decides on 0. } \\
      $\ \mathtt{V} := \forall r. \kappa(\textsf{d0}, r) \leq 0$ \quad~&  $\LTL(\mathtt{V})\ = \textbf{G}\bigl(\flocal(\textsf{d0}) \leq 0\bigr)$ \\[2mm]
     \multicolumn{2}{l}{
\textbf{Termination}: Eventually, all correct processes decide. } \\
      $\ \mathtt{T} :=\neg (\sum_{r} \kappa(\textsf{d0}, r) + \kappa(\textsf{d1}, r) \leq \texttt{N}_{\texttt{c} }-1)$ \quad~&  $\LTL(\mathtt{T})\ = \neg \textbf{G}\bigl(\cumul(\textsf{d0}) + \cumul(\textsf{d1}) \leq \texttt{N}_{\texttt{c}}\bigr)$ \\[2mm]
       \multicolumn{2}{l}{
\textbf{Restricted Termination}: If one correct process decides, then all correct processes decide.  } \\
      $\ \mathtt{RT}:=\neg (\sum_{r} \kappa(\textsf{d0}, r) + \kappa(\textsf{d1}, r) \leq 0) \implies \mathtt{T}$ \quad~&  $\LTL(\mathtt{RT})\ = 
      \neg \textbf{G}\bigl(\cumul(\textsf{d0}) + \cumul(\textsf{d1}) \leq 0\bigr) \ \implies\ \LTL(\mathtt{T})$ \\[2mm]
 \multicolumn{2}{l}{  
\textbf{Leader Uniqueness}:  At most one leader process per round.} \\
      $\ \mathtt{LU}:= \forall r. \kappa(\textsf{ldr}, r) \leq 1$ \quad~&  {$\LTL(\mathtt{LU})\ = \textbf{G}\bigl(\flocal(\textsf{ldr}) \leq 1\bigr)$} \\
\end{tabular}
}
\end{center}
\caption{Expressing correctness properties of $\ARDAs$ in $\HSCL$ and their $\LTL$ translations (see \S~\ref{subsec:LTL}). 
  $\mathcal{X} = \{\textsf{d0}, \textsf{d1}, \textsf{ldr}\}$ 
  denote decision-0, decision-1, and leader locations,
  respectively, and $\mathcal{Y}= \{\texttt{N}_c\}$ denotes the
   number of correct processes. Formulas for (restricted) termination assume that processes halt after making a decision. 
   }
\label{fig:prop-expressivity}
\end{figure}

$\HSCL$ can express the five standard correctness properties listed in Fig.~\ref{fig:prop-expressivity}. 
The same figure gives the corresponding $\HSCL$ formulas and their $\LTL$ translations.  We discuss the translation from $\HSCL$ to $\LTL$ in \S~\ref{subsec:LTL}. As we will see, $\HSCL$ formulas infact translate to restricted $\LTL$ properties, namely obligation properties~\cite{DBLP:conf/podc/MannaP89}.
These five properties guided the design of $\HSCL$, whose syntax is tailored to encode them effectively. Although $\HSCL$ is intentionally minimal, our reduction to a finite-counter system applies more broadly to arbitrary $\HSC$ properties. Thus, the logic is not intended as a general-purpose specification language, but rather as a concise intermediate representation that facilitates translation to $\LTL$.

\subsection{Verification Problems}\label{subsec:Verif-Problems} 

We can now define the four verification problems studied in this work. They are obtained by combining a choice of semantics, fixed-instance or parameterized, with a choice of specification formalism: either an arbitrary $\HSC$ property or a property expressed in $\HSCL$.

\begin{description}
    \item[P1 Fixed-instance $\HSC$-Verification] Given a template $\mathcal{T}$, an $\HSC$ property $\mathcal{P}$, and a parameter valuation $\nu$ with $\nu \models \mathsf{rc}$, does $\RDTS(\mathcal{T},\nu)\models\mathcal{P}$?
    \item[P2 Parameterized $\HSC$-Verification] Given a template $\mathcal{T}$ and an $\HSC$ property $\mathcal{P}$, does $\PRDTS(\mathcal{T})\models\mathcal{P}$?
    \item[P3 Fixed-instance $\HSCL$-Verification] Given a template $\mathcal{T}$, an $\HSCL$ formula $\varphi$, and a parameter valuation $\nu$ with $\nu \models \mathsf{rc}$, does $\RDTS(\mathcal{T},\nu)\models \varphi$?
    \item[P4 Parameterized $\HSCL$-Verification] Given a template $\mathcal{T}$ and an $\HSCL$ formula $\varphi$, does $\PRDTS(\mathcal{T})\models\varphi$?
\end{description}

Note that P$1$ and P$2$ are mathematical problems, whereas P$3$ and P$4$ are computational decision problems. 
The reductions developed in the next section relate these four problems to model checking problems over finite-state, one-counter, and finite-counter systems; see Table~\ref{tab:reductions}. In particular, the reduction for P$3$ yields $\LTL$ model checking over a finite-state system, and is therefore decidable. By contrast, the reduction for P$4$ yields $\LTL$ model checking over a finite-counter system.
Here, a finite-counter system \footnote{finite-counter systems are instances of Presburger counter systems~\cite{DBLP:journals/jancl/DemriFGD10} with quantifier-free predicates}, is a symbolic representation of a transition system whose states are valuations of finitely many variables, called \emph{counters}, ranging over the non-negative integers, and whose set of initial states and transition relation are definable by linear-arithmetic formulas over these counters. 
For model checking over such systems, we allow \(\LTL\) atoms to be linear-arithmetic formulas over the counters; see
App.~\ref{app:LTL-CS} for formal definitions.

$\LTL$ model checking over finite-counter systems is undecidable in general, thus the reduction above does not by itself give a decision procedure for P\(4\). We next show that P\(4\) is indeed undecidable, by extending the undecidability result of~\cite{DBLP:conf/tacas/StoilkovskaKWZ19} from synchronous round-based distributed algorithms to our asynchronous setting.

\begin{table}[bp]
  \centering
  \scriptsize
  \setlength{\tabcolsep}{5pt}
  \renewcommand{\arraystretch}{1.2}
  \caption{Reductions to model-checking (MC) problems over one-counter, finite-counter, and finite-state systems. All reductions are sound and complete.}
  \label{tab:reductions}
  \begin{tabular}{@{}ll c@{}}
    \toprule
    \textbf{Problem} & \textbf{Reduced problem} & \textbf{Ref.} \\
    \midrule
    P1 Fixed-instance $\HSC$-Verification 
      & MC of one-counter systems w.r.t. $\HSC$ 
      & Thm.~\ref{thm:FI-HSC} \\

    P2 Parameterized $\HSC$-Verification 
      & MC of finite-counter systems w.r.t. $\HSC$ 
      & Thm.~\ref{thm:P-HSC} \\

    P3 Fixed-instance $\HSCL$-Verification 
      &  MC of finite-state systems w.r.t. $\LTL$
      & Thm.~\ref{thm:FI-HSCL} \\

    P4 Parameterized $\HSCL$-Verification 
      & MC of finite-counter systems w.r.t. $\LTL$
      & Thm.~\ref{thm:P-HSCL} \\
    \bottomrule
  \end{tabular}
\end{table}

\begin{theorem} \label{thm:undec}
Parameterized $\HSCL$-Verification is undecidable.
\end{theorem}

\begin{proof}
We reduce from the non-halting problem for two-counter machines (known to be undecidable~\cite{PMCP_book}). %
Let $\mathcal C$ be a two-counter machine with a finite set of control states $S$, an initial state $s_{\mathsf{init}}$, a halting state $s_{\mathsf{halt}}$, and a  finite set of commands $\Delta$. Each command in $\Delta$ has the form $(s,\mathsf{op},s')$, where $\mathsf{op}$ is an operation of the form $\mathsf{inc}(c_i)$, $\mathsf{dec}(c_i)$, or $\mathsf{iszero}(c_i)$ for $i\in\{1,2\}$.
Commands with operations $\mathsf{dec}(c_i)$ and $\mathsf{iszero}(c_i)$ are enabled only when $c_i>0$ and $c_i=0$, respectively.
From $\mathcal C$ we construct a round-based process template $\mathcal T_{\mathcal C}$ and an $\HSCL$ formula $\varphi_{\mathcal C}$ such that $\mathcal C$ has no halting execution iff $\PRDTS(\mathcal T_{\mathcal C}) \models \varphi_{\mathcal C}$.

The construction maintains the following invariant. The $r$-th configuration $(s,v_1,v_2)$ of an execution of $\mathcal C$ is represented by a global configuration in which all processes have just jumped to round $r$, with one process, referred to as the controller, in location $l_s$, and $v_i$ many processes in location $l_{c_i}$ (for $i = 1,2$), and the rest of the processes are in $l_{\mathsf{res}}$ (``reserve''). Thus, initially, one process is in $l_{s_{\mathsf{init}}}$ and all other processes are in $l_{\mathsf{res}}$. 

A transition in $\mathcal C$ from configuration $(s,v_1,v_2)$ on a command, say, $(s,\mathsf{inc}(c_1),s')$, is simulated as follows. First, the controller moves from $l_s$ to an intermediate location $l_{(s,\mathsf{inc}(c_1),s')}$ in the same round and broadcasts message $m_{\mathsf{inc},c_1}$; this message enables a process to move from $l_{\mathsf{res}}$ to another intermediate location $l_{\mathsf{inc},c_1}$ of the same round and broadcast the message $m_{\mathsf{ack}}$; this acknowledgment then allows all processes to move to their corresponding locations in the next round, i.e., the controller to $l_{s'}$, the process in $l_{\mathsf{inc},c_1}$ to $l_{c_1}$, and the remaining processes to stay in their current locations. Upon entering the new round, each non-controller process broadcasts a ready message indicating whether it is in $l_{c_1}$, $l_{c_2}$, or $l_{\mathsf{res}}$. 
Once the controller enters the next round and receives all $n-1$ ready messages, the simulated transition is complete. The other commands are simulated similarly; in particular, the ready-messages counts are used to test the enabledness of decrement and zero-test commands.

The template semantics is an over-approximation of the intended simulation.
For example, several processes in $l_{\mathsf{res}}$ may receive $m_{\mathsf{inc},c_1}$ and move to $l_{\mathsf{inc},c_1}$ in the same round, thereby increasing the value of $c_1$ by more than one.
The formula $\varphi_{\mathcal C}:=\neg\psi_{\mathcal C}$ where $\psi_{\mathcal C}$ conjuncts conditions ensuring well-formedness of the simulation with a halting condition: in each round, at most one process may visit some location $l_s$ with $s\in S$, at most one process may visit an intermediate location $l_{\mathsf{inc},c_i}$ or $l_{\mathsf{dec},c_i}$, and some process eventually visits $l_{s_{\mathsf{halt}}}$.
Thus, $\varphi_{\mathcal C}$ holds on all executions of $\PRDTS(\mathcal T_{\mathcal C})$ iff no well-formed simulated execution of $\mathcal C$ reaches $s_{\mathsf{halt}}$.

The full definition of $T_\mathcal{C}$ and $\varphi_{\mathcal C}$ is given in App.~\ref{app:undecide}.
\end{proof}

\section{Reductions} %
\label{sec:abstractions}
This section presents the reductions summarized in Table~\ref{tab:reductions}. 
They are obtained by a sequence of sound and complete transformation steps. 
Due to space constraints, we present only the intuition behind each step and defer the formal definitions and proofs to the App.~\ref{app:abstractions}. %

\subsection{Reductions for \textsf{HSC} Properties}\label{subsec:CS}
In this subsection, we present a reduction pipeline for $\HSC$ properties, shown in Fig.~\ref{fig:abstractions}. The pipeline is sound and complete for $\HSC$ properties, and reduces fixed-instance semantics to one-counter systems and parameterized semantics to finite-counter systems.  At a high level, the counters of the reduced systems record, for each relevant round, (i) the number of processes at each location and (ii) the number of messages of each type that have been broadcast. For the fixed-instance semantics, given a template $\mathcal{T}$ and valuation $\nu$, the system $\RDTS(\mathcal{T},\nu)$ reduces to a finite-counter system in which all counters except one are bounded by the number of processes $\nu(n)$; hence, the target is a \emph{one-counter system}. 
For the parameterized semantics, we exploit the fact that our reductions for fixed instances guarantee that the set of counters is identical across parameter valuations. 
Taking the disjoint union over all admissible $\nu$ thus yields a \emph{finite-counter system}. %
The reduction pipeline consists of four steps. At each step $X$, we denote by $\RDTS^X(\mathcal{T}, \nu) = \langle S^X_\nu,I^X_\nu,Act^X_\nu,Tr^X_\nu\rangle$ the fixed-instance semantics after the step $X$.

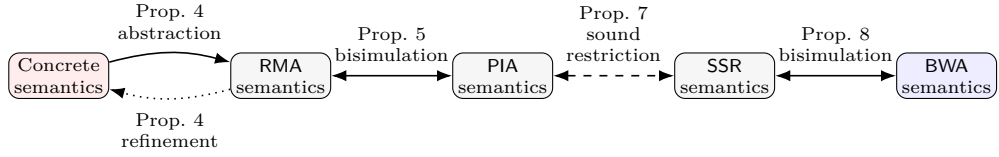
\begin{figure}[ht]
\centering
\begin{tikzpicture}[
  >=Latex,
  font=\scriptsize,
  node distance=11mm and 16mm,
  box/.style={draw, rounded corners, fill=gray!7, inner sep=2pt,
              minimum width=12mm, minimum height=6mm, align=center},
  box1/.style={draw, rounded corners, fill=red!7, inner sep=2pt,
              minimum width=12mm, minimum height=6mm, align=center},
  box2/.style={draw, rounded corners, fill=blue!7, inner sep=2pt,
              minimum width=12mm, minimum height=6mm, align=center},
  arrAbstr/.style={-Latex, semithick},             
  arrEquiv/.style={<->, semithick},                 
  arrRefine/.style={-Latex, semithick, dotted},     
  arrRestr/.style={<->, semithick, dashed}         
]

\node[box1] (rdts) {Concrete\\ semantics};
\node[box, right=of rdts] (rma) {${\RMA}$\\ semantics};
\node[box, right=of rma] (pia) {${\PIA}$\\ semantics};

\node[box, right=of pia] (ssr) {${\SSR}$\\ semantics};
\node[box2, right=of ssr] (bwa) {${\BWA}$\\ semantics};

\draw[arrAbstr, bend left=15] (rdts) to
  node[midway, above=0.6mm, align=center]{Prop.~\ref{cor:RMA-sound-complete}\\{abstraction}}  (rma);

\draw[arrRefine, bend left=15] (rma) to
  node[midway, below=0.6mm, align=center]{Prop.~\ref{cor:RMA-sound-complete}\\refinement}
  (rdts);

\draw[arrEquiv] (rma) -- (pia)
  node[midway, above=0.6mm, align=center]{Prop.~\ref{lem:PIA-completeness}\\bisimulation };

\draw[arrRestr] (pia) -- (ssr)
  node[midway, above=0.6mm, align=center]{Prop.~\ref{lem:SSR-completeness}\\sound \\restriction};

\draw[arrEquiv] (ssr) -- (bwa)
  node[midway, above=0.6mm, align=center]{Prop.~\ref{lem:BWA-completeness}\\bisimulation};

\end{tikzpicture}
\caption{Reduction pipeline for fixed-instance and parameterized verification of $\HSC$ properties. All steps are sound and complete. $\RMA$ stands for received-message abstraction, $\PIA$ for process-identity abstraction, $\SSR$ for semi-synchronous restriction, and $\BWA$ for bounded-window abstraction. In the fixed-instance (resp. parameterized) setting, the target semantics $\BWA$ is a one-counter (resp. finite-counter) system.}
\label{fig:abstractions}
\end{figure}

\subsubsection*{Step 1. Received-Message Abstraction ($\RMA$)} 
The first step exploits the fact that the \textsf{hsc}-trace of an execution is insensitive to \(\textsf{Receive}\) transitions. 
Consequently, receive transitions can be treated as stutter steps. 
In the abstract system $\RDTS^\RMA(\mathcal{T}, \nu)$, the receive transitions are dropped and 
the elements of $S^\RMA_\nu$ remain pairs of process and network states, except that the process state 
is simplified to $\PState : \llbracket 1, \nu(n) \rrbracket \rightarrow \langle \mathsf{loc}:\mathcal{L}, \mathsf{rd}:\mathbb{N} \rangle$, omitting $\mathsf{rcvMsg}$. 
Since $\Update$ transitions in $\RDTS(\mathcal{T}, \nu)$ are enabled only when their guards are justified by received messages, enabledness must now be redefined: in the abstract system, guards are evaluated \emph{existentially} with respect to subsets of messages recorded in the network state~$\NState$. 
While abstraction of received messages has appeared in prior work~\cite{RMA-TA-ATVA20,bertrand_2021} in a state-based form, we give here an action-based formulation that treats receive transitions as stuttering steps. This yields refinement and forward simulation in the sense of Lynch et al.~\cite{simulations}.%

With $\textsf{Receive}$ transitions as stutter steps, we establish a forward simulation \cite[\S~3.2]{simulations} between the concrete semantics and the $\RMA$ semantics, proving that $\RMA$ is sound for $\HSC$ properties. 
The completeness of $\RMA$ is more subtle. The abstraction discards the concrete record of which messages a process has already \emph{received and used} to satisfy previous guards. 
As a result, in the abstract enabledness relation, the witness for one update may be $R_1 \le \NState$, while the witness for the very next update may be a different $R_2 \le \NState$, without respecting the concrete monotonicity constraint that the second guard must hold on a super-set of messages already used (i.e., $R_2 \supseteq R_1$). Consequently, the abstraction may suggest that two updates can be carried out consecutively, even though no message-delivery schedule in the concrete system can realize them.  Condition~$3$ in Def.~\ref{def:RPT} is imposed to exclude this phenomenon: once a process takes a type-$0$ rule, increasing the message counts used to justify its guard cannot disable any rule that may be taken later in the same round. Under this condition, the abstract witnesses for future updates can always be arranged consistently with a concrete receive schedule. Thus, with \textsf{Receive} transitions as stutter steps, we establish a refinement \cite[\S~3.1]{simulations} between the $\RMA$ semantics to the concrete semantics, proving that $\RMA$ is complete for $\HSC$ properties.

\begin{proposition}[Soundness and Completeness of $\RMA$]\label{cor:RMA-sound-complete} 
For every template $\mathcal{T}$ and parameter valuation $\nu$: $\mathsf{Traces}(\RDTS(\mathcal{T}, \nu)) =\mathsf{Traces}(\RDTS^\RMA(\mathcal{T}, \nu))$.
\end{proposition}

\subsubsection*{Step 2. Process-Identity Abstraction ($\PIA$)}
The second step exploits the fact that the \(\hsc\)-trace of an execution is insensitive to the identities of the processes performing \(\Update\) transitions. 
Consequently, process identities can be dropped, resulting in an abstraction in which the exact states of processes are no longer recorded, only their counts.
In the resulting abstract system $\RDTS^\PIA(\mathcal{T}, \nu)$, the elements of $S^\PIA_\nu$ remain pairs of process and network states, except that the process state is simplified to $\PState : \mathcal{L} \times \mathbb{N} \rightarrow \mathbb{N}$.
Here $\PState(\ell, r)$ denote the number of processes in location $\ell$ of round $r$.
Elements of $Act^\PIA_\nu$ are of the form $\Update(\rho, r)$, again omitting process identities.
Since the seminal paper~\cite{GS-jacm92}, this counting abstraction is classical in the verification of systems composed of identical anonymous processes. The process-identity abstraction induces an action-based bisimulation~\cite[Def.~7.15]{book:PMC} between $\RDTS^{\RMA}(\mathcal{T}, \nu)$ and $\RDTS^{\PIA}(\mathcal{T}, \nu)$, which entails the following result:

\begin{proposition}
[Soundness and Completeness of $\PIA$]\label{lem:PIA-completeness}
For every template $\mathcal{T}$ and parameter valuation $\nu$,
\(\mathsf{Traces}(\RDTS^\PIA(\mathcal{T}, \nu)) = \mathsf{Traces}(\RDTS^\RMA(\mathcal{T}, \nu)).\)
\end{proposition}

The system $\RDTS^{\PIA}(\mathcal{T},\nu)$ is a counter system with \emph{infinitely many} counters. 
Indeed, each round $r \in \mathbb{N}$ introduces $|\mathcal{L}|+|\mathcal{M}|$ counters: a location counter $\PState(\ell,r)$ for every $\ell \in \mathcal{L}$ and a message counter $\NState(m,r)$ for every $m \in \mathcal{M}$. 
To obtain a finite-counter reduction, a natural question is whether all counters must be tracked at all times. 
Intuitively, counters become redundant once they can no longer influence future updates.
To capture this, we define the \emph{relevant window} of a configuration as the interval of round tags whose counters may be nonzero and can affect future updates. 
Under the $\PIA$ semantics, for a configuration $\cfg$, this window spans from the smallest ($\rmin(\cfg)$) to the largest ($\rmax(\cfg)$) round occupied by a process in $\cfg$.
A uniformly bounded relevant window across all configurations enables a direct reduction to a finite number of counters, as discussed in Step~4.
However, due to asynchrony, the \emph{round drift} $\rmax(\cfg) - \rmin(\cfg)$ is unbounded across configurations of $\RDTS^{\PIA}(\mathcal{T},\nu)$, preventing such a reduction. 
To this end, our next step constrains the enabledness of updates to enforce a uniformly bounded relevant window while preserving soundness for $\HSC$ properties.

\subsubsection*{Step 3. Semi-Synchronous Restriction ($\SSR$)}
The third step exploits the fact that the $\hsc$-trace of an execution is insensitive to the order of $\Update$ transitions. 
Consequently, verification can be restricted to a subset of executions of $\RDTS^\PIA(\mathcal{T}, \nu)$ such that every execution has a counterpart within this subset that executes the same multiset of $\Update$ actions. 
To achieve this, we impose a \emph{semi-synchronous restriction} (\SSR), which permits an action $a = \Update(\rho, r)$ from a configuration $\cfg$ only if its target round $\tgt(a) = r + \rho.\type$ is at least $\rmax(\cfg)$, i.e., the highest round currently occupied by a process in $\cfg$, also called the \emph{frontier round} of $\cfg$. 
The $\SSR$ therefore restricts executions by forbidding transitions that would move a process below the frontier round. 
As a result, since round jumps are bounded by $b$, the set of relevant counters in a configuration $\cfg$ is confined to the interval $\llbracket \rmax(\cfg) - b,\, \rmax(\cfg)\rrbracket$, which has a uniform size of $b{+}1$ across all configurations. To establish the soundness of $\SSR$ for $\HSC$ properties, we use commutativity arguments similar to the one used in~\cite{reduction-thm}.

An execution $\pi=\cfg_1\,a_1\,\cfg_2\,a_2\cdots$ of \(\RDTS^\PIA(\mathcal{T}, \nu)\) is called \emph{semi-syn\-chronous} iff its target sequence \(\tgt(a_1)\tgt(a_2)\dots\) is nondecreasing. This notion precisely characterizes the executions of $\RDTS^{\SSR}(\mathcal{T},\nu)$. 
The equivalence follows from the observation that, for every transition $(\cfg,a,\cfg')\in Tr_\nu^\PIA$, the frontier round  evolves as $\rmax(\cfg')=\max(\rmax(\cfg),\tgt(a))$. 

\begin{lemma}
\label{lm:SSR-soundness} 
In $\RDTS^{\PIA}(\mathcal{T},\nu)$, for every execution $\pi$ there exists a semi-syn\-chronous execution $\pi'$ that starts from the same initial configuration as $\pi$ and contains the same multiset of update actions.
\end{lemma}

The proof proceeds by transforming any execution $\pi$ into a semi-synchronous execution through adjacent swaps of out-of-order update actions, thereby eliminating inversions in the target-round sequence. 
Each swap is justified by the \emph{commutativity property}: two consecutive updates commute whenever the source round of the later differs from the target round of the earlier. 
To ensure convergence, swaps are performed in a disciplined manner that gives a monotone sequence of prefixes converging to a well-defined limiting execution,  which is semi-synchronous.

The same commutativity property gives a stronger result that further reduces nondeterminism in semi-synchronous executions, which we exploit in the experiments reported in \S~\ref{sec:evaluation}.
For every execution, there exists a semi-synchronous execution (with the same initial configuration and the same multiset of update actions) with the additional property that we call \emph{strong}: for each round $r$, all \emph{jump updates} targeting $r$ (whose source round is smaller than $r$) happen before all \emph{local updates} targeting $r$ (whose source round equals $r$).
Moreover, all jump updates are executed as a multiset in a fixed order which can alternatively be viewed as occurring synchronously in a single step.

\begin{proposition}[Soundness and Completeness of $\SSR$]\label{lem:SSR-completeness} 
For every template $\mathcal{T}$ and parameter valuation $\nu$,
\(\mathsf{Traces}(\RDTS^\SSR(\mathcal{T}, \nu)) = \mathsf{Traces}(\RDTS^\PIA(\mathcal{T}, \nu)).\)
\end{proposition}

\subsubsection*{Step 4. Bounded-Window Abstraction ($\BWA$)}
The fourth abstraction step exploits the observation that configurations need not maintain counters that can no longer influence future updates.
Since in $\RDTS^\SSR(\mathcal{T}, \nu)$ the size of the relevant window is bounded by $b+1$ across all configurations, the $\BWA$ forgets counters outside this window and encodes the relevant portion as a sliding window of width $b+1$. 
In the resulting abstract system $\RDTS^\BWA(\mathcal{T}, \nu)$, each configuration in $S_\nu^\BWA$ is a triple consisting of a frontier round $\rmax \in \mathbb{N}$, a process state $\PState: \mathcal{L} \times \llbracket 0, b\rrbracket \rightarrow\mathbb{N}$ and a network state $\NState: \mathcal{M} \times \llbracket 0, b\rrbracket \rightarrow\mathbb{N}$. 
The process and network states record counters only for rounds ranging from the current frontier round down to $b$ rounds below.
The $\BWA$ induces an \emph{action-based bisimulation} \cite[Def.~7.15]{book:PMC} between $\RDTS^{\SSR}(\mathcal{T},\nu)$ and $\RDTS^{\BWA}(\mathcal{T},\nu)$, establishing the following result:

\begin{proposition}[Soundness and Completeness of $\BWA$]\label{lem:BWA-completeness} 
For every template $\mathcal{T}$ and parameter valuation $\nu$,
\(\mathsf{Traces}(\RDTS^\BWA(\mathcal{T}, \nu)) = \mathsf{Traces}(\RDTS^\SSR(\mathcal{T}, \nu)).\)
\end{proposition}

The system $\RDTS^{\BWA}(\mathcal{T},\nu)$ is a finite-counter system.  %
Indeed, it maintains a single frontier-round counter $\rmax$, and for each depth $d \in \llbracket 0, b \rrbracket$, there are $|\mathcal{L}| + |\mathcal{M}|$ counters: a location counter $\PState(\ell,d)$ for every $\ell \in \mathcal{L}$ and a message counter $\NState(m,d)$ for every $m \in \mathcal{M}$. 
Moreover, $\RDTS^{\BWA}(\mathcal{T},\nu)$ can be viewed as a \emph{one-counter system} with a \emph{nondecreasing} counter. This is because the total number of processes $\nu(n)$ is bounded, and each process can send only a bounded number of messages ($\nu(n) \times |\mathcal{L}|$) per round. 
Consequently, all location and message counters are bounded, and the only unbounded counter is $\rmax$, which increases monotonically. 
Thus, from Prop.~\ref{cor:RMA-sound-complete},~\ref{lem:PIA-completeness},~\ref{lem:SSR-completeness}, and~\ref{lem:BWA-completeness}, we have:

\begin{theorem}\label{thm:FI-HSC}
For every template $\mathcal{T}$, parameter valuation $\nu$, and $\HSC$ property $\mathcal{P}$, the system $\RDTS^\BWA(\mathcal{T}, \nu)$ is a one-counter system with a nondecreasing counter $\rmax$, and $\RDTS^\BWA(\mathcal{T}, \nu) \models \mathcal{P}$  iff
$\RDTS(\mathcal{T}, \nu) \models \mathcal{P}$.

\end{theorem}

We define the parameterized $\BWA$ semantics, denoted $\PRDTS^\BWA(\mathcal{T})$, as the disjoint union of all fixed-instance systems $\RDTS^\BWA(\mathcal{T}, \nu)$ over admissible parameter valuations $\nu$.
Since the set of counter variables in $\RDTS^\BWA(\mathcal{T}, \nu)$ is independent of $\nu$, the states of $\PRDTS^\BWA(\mathcal{T})$ can be represented using finitely many counters: the frontier-round counter; the location and message counters inherited from $\RDTS^\BWA(\mathcal{T}, \nu)$; and an additional set of $|\mathsf{P}|$ counters encoding the parameter valuation~$\nu$.
Note, however, that the location and message counters in parameterized states need not be bounded, since the parameter $\nu(n)$ itself may range over unbounded values.
Recall from \S~\ref{subsubsec:HSC} that the parameterized semantics satisfies an $\HSC$ property iff all its admissible fixed-instance semantics do. Hence, from Thm.~\ref{thm:FI-HSC}, we have:

\begin{theorem}\label{thm:P-HSC}
For every template $\mathcal{T}$ and $\HSC$ property $\mathcal{P}$, the  system  $\PRDTS^\BWA(\mathcal{T})$ is a finite-counter system, and $\PRDTS^\BWA(\mathcal{T}) \models \mathcal{P}$ iff $\PRDTS(\mathcal{T}) \models \mathcal{P}$.
\end{theorem}

\subsection{Further Reductions for \textsf{HSCL} Formulas}\label{subsec:LTL}
We now present a further reduction pipeline for $\HSCL$ formulas, shown in Fig.~\ref{fig:hscl-to-ltl}. %

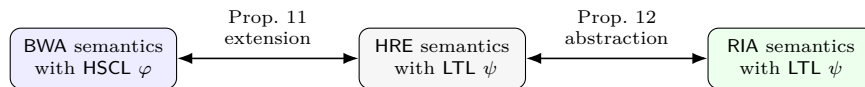
\begin{figure}[ht]
\centering
\begin{tikzpicture}[
  >=Latex,
  font=\scriptsize,
  node distance=11mm and 24mm,
  box/.style={draw, rounded corners, fill=gray!7, inner sep=2pt,
              minimum width=22mm, minimum height=8mm, align=center},
  box2/.style={draw, rounded corners, fill=blue!7, inner sep=2pt,
              minimum width=22mm, minimum height=8mm, align=center},
   box3/.style={draw, rounded corners, fill=green!7, inner sep=2pt,
              minimum width=22mm, minimum height=8mm, align=center},
  arrAbstr/.style={-Latex, semithick},              
  arrEquiv/.style={<->, semithick},                 
]

\node[box2] (bwa) {${\BWA}$ semantics\\{with $\HSCL$ $\varphi$}};
\node[box, right=of bwa] (hr) {${\mathsf{HRE}}$ semantics\\{ with $\LTL$ $\psi$}};
\node[box3, right=of hr] (ra) {$\RIA$ semantics \\{with $\LTL$ $\psi$}};

\draw[arrEquiv] (bwa) -- (hr)
  node[midway, above=0.6mm, align=center]
  {Prop.~\ref{lem:LTL-transaltion}\\ extension};

\draw[arrEquiv] (hr) -- (ra)
  node[midway, above=0.6mm, align=center]
  {Prop.~\ref{lem:RIA-completeness}\\abstraction};
  
\end{tikzpicture}
\caption{
Further reduction pipeline for fixed-instance and parameterized verification of $\HSCL$ formulas. 
All steps are sound and complete. 
$\BWA$ stands for bounded-window abstraction, $\HRE$ for history-record extension, and $\RIA$ for round-identity abstraction. 
In the fixed-instance (resp. parameterized) setting, the target semantics $\RIA$ is a finite-state (resp. finite-counter) system.
}
\label{fig:hscl-to-ltl}
\end{figure}

\subsubsection*{Step 5. History-Record Extension ($\HRE$)} The fifth step evaluates \(\HSCL\) formulas from the state point of view by extending each \(\BWA\) configuration with a ``history record'' that retains exactly those fragments of the execution history that are relevant to \(\HSCL\).
Given an \(\HSCL\) formula $\varphi$ over locations $\mathcal{X}$ and parameters $\mathcal{Y}$, we construct an extended system $\RDTS^\HRE(\mathcal{T},\nu,\varphi)$ by augmenting every state in \(S_\nu^\BWA\) with two counters for each $\ell \in \mathcal{X}$:
a cumulative counter \(\cumul(\ell)\), recording the number of visits to $\ell$ across all rounds; and
a local counter \(\local(\ell)\), recording the number of visits to $\ell$ in the current frontier round. Only the frontier round requires a dedicated local counter, since under \(\BWA\) no update targets a round strictly below the frontier round. Thus, once a round is left behind, its local counts of process visits become fixed. 

Observe that local counters are bounded by the total number of processes.
In contrast, cumulative counters may grow unboundedly. However, for a fixed parameter valuation $\nu$ and formula $\varphi$, each cumulative counter can be capped at one more than the largest threshold appearing in the cumulative atoms of $\varphi$.
Counts beyond this bound are irrelevant to satisfaction of \(\HSCL\) formulas and can therefore be safely truncated. 
As a result, the system $\RDTS^\HRE(\mathcal{T},\nu, \varphi)$ remains a single-counter system with a nondecreasing counter $\rmax$. 

Next, we give a recursive procedure to translate the $\HSCL$ formula $\varphi$ into an $\LTL$ formula $\LTL(\varphi)$.  We use a variant of the standard $\LTL$ syntax in which the only temporal operator is "globally", written $\mathbf{G}$, and atomic propositions are linear-arithmetic formulas over the counter variables;  see App.~\ref{app:LTL-CS} for the formal syntax and semantics in our setting. 
For universal round local atoms $\forall_r.\alpha_r$ where  $\alpha_r = \sum_{\ell\in\mathcal X} c_\ell \cdot \kappa(\ell,r) \le t$, define the translation $\LTL(\forall_r.\alpha_r) = \mathbf{G}(\sum_{\ell \in \mathcal{X}} c_\ell \cdot \mathsf{local}(\ell) \leq t)$. For cumulative atoms $\beta = \sum_{\ell\in \mathcal{X}} c_\ell \cdot \sum_r \kappa(\ell,r) \le t$, define the translation $\LTL(\beta) = \mathbf{G}(\sum_{\ell \in \mathcal{X}} c_\ell \cdot \mathsf{cumul}(\ell) \leq t)$. The Boolean cases are handled compositionally: $\LTL(\neg \varphi) = \neg \LTL(\varphi)$ and $\LTL(\varphi \wedge \varphi') = \LTL(\varphi) \wedge \LTL(\varphi')$. %
Fig.~\ref{fig:prop-expressivity} provides $\LTL$ translations for core $\HSCL$ formulas. %

\begin{proposition}[Soundness and Completeness of $\HRE$]
\label{lem:LTL-transaltion}
    For every $\HSCL$ formula $\varphi$, 
    \(\RDTS^\HRE(\mathcal{T},\nu, \varphi)\models \LTL(\varphi)\text{ iff }\RDTS^\BWA(\mathcal{T}, \nu) \models \varphi.\)
\end{proposition}

\subsubsection*{Step 6. Round-Identity Abstraction ($\RIA$)} The sixth and final step exploits the fact that the properties expressible in $\HSCL$ are insensitive to the round identifiers. Consequently, the frontier-round counter $r_{\max}$ can be dropped from $\HRE$ configurations.
The resulting abstract system $\RDTS^\RIA(\mathcal{T}, \nu, \varphi)$ induces a \emph{state-based bisimulation}  \cite[Def.~7.7]{book:PMC} with respect to $\RDTS^\HRE(\mathcal{T}, \nu, \varphi)$, giving:%

\begin{proposition}[Soundness and Completeness of $\RIA$]
\label{lem:RIA-completeness}
    For every $\HSCL$ formula $\varphi$,\\ $\RDTS^\RIA(\mathcal{T}, \nu,\varphi) \models \LTL(\varphi)$ iff $ \RDTS^\HRE(\mathcal{T}, \nu,\varphi) \models \LTL(\varphi)$.
\end{proposition}

Note that the system $\RDTS^\RIA(\mathcal{T}, \nu, \varphi)$ is a finite-counter system, all of whose counters are bounded by a finite value determined by the parameter valuation $\nu$, and is therefore a finite-state system. By Thm.~\ref{thm:FI-HSC} and Prop.~\ref{lem:LTL-transaltion} and~\ref{lem:RIA-completeness}, we have:

\begin{theorem}\label{thm:FI-HSCL}
For every template $\mathcal{T}$, parameter valuation $\nu$, and an $\HSCL$ formula $\varphi$, 
the system $\RDTS^\RIA(\mathcal{T}, \nu, \varphi)$ is a finite-state system, and 
$\RDTS^\RIA(\mathcal{T}, \nu, \varphi) \models \LTL(\varphi) $ iff $ \RDTS(\mathcal{T}, \nu) \models \varphi$.
\end{theorem}

The parameterized $\RIA$ semantics, denoted $\PRDTS^\RIA(\mathcal{T}, \varphi)$, is defined as the disjoint union of all fixed-instance systems $\RDTS^\RIA(\mathcal{T}, \nu, \varphi)$ over all admissible parameter valuations $\nu$. Since the set of counters in $\RDTS^\RIA(\mathcal{T}, \nu, \varphi)$ does not depend on $\nu$, each state of $\PRDTS^\RIA(\mathcal{T}, \varphi)$ can be represented with finitely many counters: those of $\RDTS^\RIA(\mathcal{T}, \nu, \varphi)$ together with an additional $|\mathsf{P}|$ counters encoding the parameter valuation $\nu$. Hence, $\PRDTS^\RIA(\mathcal{T}, \varphi)$ is a finite-counter system. The atoms of the $\LTL$ formula $\LTL(\varphi)$ are interpreted over these parameter counters and over the local/cumulative visit counters of $\PRDTS^\RIA(\mathcal{T}, \varphi)$. 
By the definition of the disjoint union, $\PRDTS^\RIA(\mathcal{T}, \varphi) \models \LTL(\varphi)$ iff, for all $\nu \models \textsf{rc}$, $\RDTS^\RIA(\mathcal{T}, \nu, \varphi) \models \LTL(\varphi).$ Therefore, by Thm.~\ref{thm:FI-HSCL}, we obtain:

\begin{theorem}\label{thm:P-HSCL}
For every template $\mathcal{T}$ and $\HSCL$ formula $\varphi$, the system $\PRDTS^\RIA(\mathcal{T}, \varphi)$ is a finite-counter system, and $\PRDTS^\RIA(\mathcal{T}, \varphi) \models \LTL(\varphi) $ iff $ \PRDTS(\mathcal{T}) \models \varphi.$
\end{theorem}
{The above six transformation steps establish the correctness of our reduction. App.~\ref{app:counter-system} gives a direct construction of the reduced finite-state and finite-counter system,
\(\RDTS^\RIA(\mathcal{T}, \nu, \varphi)\) and
\(\PRDTS^\RIA(\mathcal{T}, \varphi)\),
from the input template \(\mathcal{T}\) and formula \(\varphi\).}

\section{Case Studies}\label{sec:evaluation} 
We evaluate the practical usefulness of our reduction on four case studies: (i) Ben-Or's consensus with crash faults~\cite[Algo A]{ben-or}, (ii) Ben-Or's consensus with Byzantine faults~\cite[Algo B]{ben-or}, (iii) Bracha's consensus with Byzantine faults~\cite[Fig. 4]{bracha}, and (iv) Raft's leader-election~\cite[Sec. 5.2]{raft}. 
The anonymous repository~\cite{AnonPaperX2025} provides source files and instructions to reproduce all experiments.  For each case study,  the repository includes an \texttt{.smv} file whose initial commented block specifies the corresponding process template, followed by the $\RIA$ semantics used as input to nuXmv. The models abstract randomization by nondeterminism; as a result, termination is no longer guaranteed, although the remaining properties are unaffected.
Byzantine faults are modeled explicitly using dedicated fault locations, from which faulty processes may nondeterministically broadcast messages of arbitrary types.
The model checker nuXmv is called via its IC3 engine (\texttt{check\_ltlspec\_ic3}) to check whether the translated $\LTL$ specifications hold on the $\RIA$ semantics. Table~\ref{tab:ic3-results} reports, for each case study and property, the verification time and the maximal IC3 depth. %
Across all experiments, IC3 reaches shallow depths and terminates in under $15$ seconds.
To stress-test our reduction, we introduced intentional faults into the algorithm models by weakening the resilience condition or altering key guard predicates. 
As expected, for these faulty variants nuXmv produced counterexamples within seconds.

\begin{table}[t]
\centering
\caption{Case-study results with the IC3 engine of nuXmv.
Columns $b$, $|\mathcal{L}|$, $|\mathsf{Rules}|$, and $\textsf{rc}$ give the jump bound, number of locations, number of rules, and resilience condition of the process template. For each property in Fig.~\ref{fig:prop-expressivity}, the rest of columns report verification time in seconds, with the maximal IC3 depth in parentheses. For \textsf{T}, the  entry reports the time and depth  to find a counterexample.
}
\label{tab:ic3-results}
\resizebox{\linewidth}{!}{
\begin{tabular}{|l||c|c|c|c||c|c|c|c|c|}
\hline
\textbf{Protocol} & b & $\mathbf{|\mathcal{L}|}$ & $\mathbf{|\mathsf{Rules}|}$ & \textsf{rc} & \textbf{A} & \textbf{V} & \textbf{RT} & $\textbf{T}$ & \textbf{LU} \\
\hline\hline
Ben-Or (Crash)     & 1 & 9 &  26  & $n > 2t$ & 1.4s (13) & 0.4s (9)  & 3.1s (8) & 0.5s (3) & -- \\
\hline
Ben-Or (Byzantine)   & 1 & 10 & 27 & $n > 5t$ & 7.0s (11) & 1.2s (7)  & 4.3s (7)  & 0.6s (3) & -- \\
\hline
Bracha (Byzantine)  & 1 & 12 & 31 & $n > 3t$ & 14.0s (14)& 1.8s (8)  & 6.5s (11) & 0.7s (3) & -- \\
\hline
Raft leader election  & 2 & 11 & 25 & $n > 2t$ & --        & --        & --        & --       & 1.8s (8) \\
\hline
\end{tabular}
}
\end{table}

\section{Conclusion}\label{sec:conclusion}
Asynchronous round-based distributed algorithms underpin many modern distributed systems, including many consensus algorithms at the core of blockchain technologies~\cite{8972381}. Our reduction offers a formal step toward their verification via mature symbolic model checkers and suggests a number of challenging extensions.

First, an important direction is to extend our reduction to incorporate probabilistic behavior and so allow verification of almost-sure termination.
While existing probabilistic model checkers~\cite{prism,storm} can efficiently verify almost-sure reachability conditions on \emph{finite} MDPs, obtaining such a model is far from immediate, even for a fixed parameter valuation. The main challenge arises in the $\SSR$ step as a probabilistic semantics requires to reorder not only a trace, but the whole execution tree.

Second, a compelling challenge is to extend our reduction to handle algorithms whose number of locations within a round grows with the round index, as is the case for variants of Paxos~\cite{Paxos,fastpaxos} and DAG-based consensus~\cite{BL-coins20,KKNS-podc21,KeidarNPS23,GagolLSS19,lachesis} (where the number of locations also grows with the number of processes).
Interestingly, the reduction we presented in this paper does not assume that the set of locations within each round is finite; however, the reduced system would then have infinitely many counters.

\bibliography{concur}

\appendix

\section{Details on Section~\ref{sec:modelling}}
\subsection{Fixed-instance semantics}\label{app:semantics-fixed}
Let $\mathcal{T} = \langle \mathsf{P}, \mathsf{rc}, \mathcal{L}, \mathcal{I}, \mathcal{M}, \mathsf{Bcast}, \mathsf{Rules} \rangle$ be a template
and $\nu : \mathsf{P} \to \mathbb{N}$ a \emph{parameter valuation}
such that $\nu \models \mathsf{rc}$.  In the system consisting of
$\nu(n)$ identical processes, each process execute the behavior
specified by
$\langle \mathcal{L}, \mathcal{I},\mathcal{M},  \mathsf{Bcast}, \mathsf{Rules}
\rangle$.

The \emph{fixed-instance semantics} of the template $\mathcal{T}$ for a parameter valuation $\nu$, denoted $\RDTS(\mathcal{T}, \nu)$, is an action-labeled transition system: %
\[
  \RDTS(\mathcal{T},\nu)
  \;=\;
  \langle S_\nu,\; I_\nu,\; Act_\nu,\; Tr_\nu\rangle ,
\]
where $S_\nu$ denotes the set of \emph{configurations}, $I_\nu \subseteq S_\nu$ the set of \emph{initial configurations}, $Act_\nu$ the \emph{action set}, and $Tr_\nu \subseteq S_\nu \times Act_\nu \times S_\nu$ the \emph{transition relation}.

\noindent \textbf{Configurations.} A \emph{configuration} is a pair $\cfg = \langle \PState, \NState \rangle$ consisting of a process state and a network state. 
The process state is a function $\PState : \llbracket 1, \nu(n) \rrbracket \rightarrow \langle \mathsf{loc}:\mathcal{L},\mathsf{rd}:\mathbb{N},\mathsf{rcvMsg}:\mathcal{M} \times \mathbb{N} \rightarrow \mathbb{N} \rangle$ %
that assigns to each process a location in $\mathcal{L}$, a round number in $\mathbb{N}$, and a multiset of received messages without sender identities. Specifically, $\mathsf{rcvMsg}(m, r)$ records how many times a message of type $m \in \mathcal{M}$ was received with round tag $r$. 
This representation abstracts away sender identities and retains only message counts --- a standard abstraction known to preserve the temporal properties typically verified for fault-tolerant distributed algorithms~\cite{counting-abs}. {The reason is that correctness properties typically depend on the control locations occupied by processes, rather than on the specific messages they send or receive. Moreover, the guards that determine these locations are threshold-based and independent of process identifiers.}
The network state $\NState : \mathcal{M} \times \mathbb{N} \rightarrow \mathbb{N}$ records how many messages of each type and round tag were broadcast, irrespective of the sender’s identity.
The acyclicity assumption in Definition~\ref{def:RPT} ensures that no process visits the same location more than once per round. Consequently, each process broadcasts at most $|\mathcal{L}|$ messages in a round, and the range of $\NState$ is therefore bounded above by $\nu(n) \times |\mathcal{L}|$. %
Moreover, if all processes are in a round with number less than $r$ (formally, $\forall i : \PState(i).\mathsf{rd} < r$), then $\NState(\cdot, r)$ must be the zero function ${\bf 0}$, mapping every argument to $0$. %
Finally, in every configuration, the number of received messages is \emph{pointwise bounded} by the number of broadcast messages:
\begin{equation}\label{equ:rec-bounded-bcast}
   \forall i,\ \forall (m,r) \in \mathcal{M} \times \mathbb{N}:\quad
\PState(i).\mathsf{rcvMsg}(m,r) \ \le\ \NState(m,r). 
\end{equation}

\noindent A configuration $\langle \PState, \NState \rangle$ is \emph{initial} if every process is in an initial location at round~0 with an empty set of received messages, and the network contains no broadcasts. Formally:
for all $i \in \llbracket 1, \nu(n) \rrbracket$, $\PState(i).\mathsf{loc} \in \mathcal{I}$, $\PState(i).\mathsf{rd} = 0$, and  $\PState(i).\mathsf{rcvMsg} = \bf 0$; and $\NState = \bf 0$.

\noindent \textbf{Actions and Transitions.} 
A process $i \in \llbracket 1, \nu(n) \rrbracket$ can perform two types of actions: (i) $\Recv(i, \langle m, r\rangle)$, which represents the reception of a message of type $m$ tagged with round~$r$; (ii) $\Update(i, \rho, r)$, which represents updating its location by applying rule $\rho \in \mathsf{Rules}$ at round $r \in \mathbb{N}$.  
These form the action set $Act_\nu$, in which each action represents a step executed by a single process.
The transition relation $Tr_\nu$ consists of all triples $(\cfg, a, \cfg')$ such that
$a \in Act_\nu$ is \emph{enabled} in the configuration $\cfg \in S_\nu$ (denoted $\enabled(\cfg, a)$), and executing $a$ at $\cfg$ yields $\cfg'$ (denoted $\effect(\cfg, a) = \cfg'$). We now formally define $\enabled$ and $\effect$ for \emph{receive} and \emph{update} actions.\\

\noindent \textsf{Receive}. \label{receive-tr} In a configuration $\cfg = \langle \PState, \NState \rangle$, the action $\Recv(i, \langle m, r \rangle)$ is enabled when process $i$ has not yet received all broadcasts of type $m$ tagged with round~$r$. Formally,
\[\enabled(\cfg, \Recv(i, \langle m, r\rangle)) := \PState(i).\mathsf{rcvMsg}(m, r) < \NState(m, r).\]
Executing $\Recv(i, \langle m, r\rangle)$ delivers one message of type $m$ tagged with round $r$ to process $i$ and leaves all other components unchanged. Formally, \\

\noindent $\effect(\cfg, \Recv(i, \langle m, r\rangle)) := \langle \PState', \NState'\rangle$, where:
\begin{itemize}
    \item For every process $j \neq i$, the state remains unchanged: $\PState'(j) = \PState(j)$.
    \item Location and round component of process~$i$ state remain unchanged:\\ $\PState'(i).\mathsf{loc} = \PState(i).\mathsf{loc}$ and $\PState'(i).\mathsf{rd} = \PState(i).\mathsf{rd}$.
    \item The receive component of process~$i$ state is incremented at $\langle m, r\rangle$:\\ $\PState'(i).\mathsf{rcvMsg}(m,r) = \PState(i).\mathsf{rcvMsg}(m, r) + 1$ and for all\\ $\langle m', r'\rangle \neq \langle m, r\rangle$, $\PState'(i).\mathsf{rcvMsg}(m',r') = \PState(i).\mathsf{rcvMsg}(m', r')$.
    \item The network state remains unchanged: $\NState' = \NState$.
\end{itemize}

\noindent \textsf{Update.} In a configuration $\cfg = \langle \PState, \NState \rangle$, the action $\Update(i, \rho, r)$ is enabled when process $i$ is in location $\rho.\mathsf{frm}$ at round~$r$, and its received messages in round $r$ satisfy the guard $\rho.\mathsf{guard}$. Formally, 
\begin{align*}
\enabled(\cfg,\Update(i, \rho, r)) :=\ 
& \PState(i).(\mathsf{loc}, \mathsf{rd}) = (\rho.\mathsf{frm}, r)\ \text{and} \\
& \PState(i).\mathsf{rcvMsg}(\_, r) \models \rho.\mathsf{guard}[\mathsf{P} \leftarrow \nu].
\end{align*}%
\noindent Executing $\Update(i, \rho, r)$ has two effects. First, process $i$ moves from location $\rho.\mathsf{frm}$ to $\rho.\mathsf{to}$ and updates its round counter to $r + \rho.\mathsf{type}$. Second, it broadcasts a message of type $\mathsf{Bcast}(\rho.\mathsf{to})$, tagged with round~$r + \rho.\mathsf{type}$. Formally,\\

\smallskip
\noindent $\effect(\cfg, \Update(i, \rho, r)) = \langle \PState', \NState'\rangle$, where:
\begin{itemize}
    \item For every process $j \neq i$, the state remains unchanged:  $\PState'(j) = \PState(j)$.
    \item The location component of process~$i$ is updated to $\rho.\mathsf{to}$:\\ $\PState'(i).\mathsf{loc} = \rho.\mathsf{to}$.
    \item The round component of process~$i$ state is incremented by $\rho.\mathsf{type}$: \\$\PState(i)'.\mathsf{rd} = r + \rho.\mathsf{type}$
    \item The receive component of process~$i$ remains unchanged:\\ $\PState'(i).\mathsf{rcvMsg}(m,r) = \PState(i).\mathsf{rcvMsg}(m, r)$.
    \item The network state is incremented at $\langle \mathsf{Bcast}(\rho.\mathsf{to}), r + \rho.\mathsf{type} \rangle$: \\$\NState'(m', r') = \NState(m', r') + 1$ if $\langle m', r' \rangle = \langle \mathsf{Bcast}(\rho.\mathsf{to}), r + \rho.\mathsf{type} \rangle$, and $\NState'(m', r') = \NState(m', r')$ otherwise.
\end{itemize}

An \emph{execution} of $\RDTS(\mathcal{T}, \nu)$ is a finite or infinite alternating sequence $\pi = \cfg_0 a_1 \cfg_1 a_2 \dots,$ such that $\cfg_0 \in I_\nu$ is an initial configuration and, for all $k$, $(\cfg_k,a_{k+1},\cfg_{k+1}) \in Tr_\nu$. {We adopt this definition of executions uniformly for all subsequent semantics and do not restate it.}

%
%

%
%
%
%
%

%

%

%
%
%

%

%
%

\subsection{Parameterized semantics}\label{app:param-sem}
The parameterized semantics is a disjoint union of all fixed-instance semantics $\RDTS(\mathcal{T},\nu)$ for $\nu \models \mathsf{rc}$.
Fix a template $\mathcal{T}$, and let $\mathsf{Val} =\{\nu: \mathsf{P} \rightarrow\mathbb{N}\ |\ \nu \models \mathsf{rc}\}$ denote the set of all admissible parameter valuations. 
The parameterized semantics of $\mathcal{T}$ is an action-labeled transition system, called the \emph{parameterized} $\RDTS$ ($\PRDTS$), defined as 
\[\PRDTS(\mathcal{T})= \langle S,\; I,\; Act,\; Tr\rangle,\]
which is the \emph{disjoint union} of $\RDTS(\mathcal{T}, \nu)$ over all $\nu \in \mathsf{Val}$, i.e.:
\begin{align*}
  &S   = \bigcup_{\nu\in \mathsf{Val}}
         \{\nu\}\times S_\nu, \quad 
  I   = \bigcup_{\nu\in\mathsf{Val}}
         \{\nu\}\times I_\nu, \quad 
  Act = \bigcup_{\nu\in\mathsf{Val}}
         \{\nu\}\times Act_\nu,\\
  &Tr = \bigl\{
           (\nu,\cfg)\xrightarrow{\,(\nu,a)\,}(\nu,\cfg')
           \;\bigm|\;
           \nu\in\mathsf{Val},\,
           (\cfg,a,\cfg')\in Tr_\nu
         \bigr\}.
\end{align*}
{For every execution $\pi = (\nu,\cfg_0)(\nu,a_1)(\nu,\cfg_1)\dots$ of $\PRDTS(\mathcal{T})$, let  $\pi_{\downarrow_\nu} = \cfg_0 a_1 \cfg_1\dots$. Clearly $\pi_{\downarrow_\nu}$ is an execution of $\RDTS(\mathcal{T},\nu)$. Moreover, for every execution $\pi$ of $\RDTS(\mathcal{T},\nu)$, there exists execution $\pi'$ of $\PRDTS(\mathcal{T})$ such that $\pi = \pi'_{\downarrow_\nu}$.}\label{def:PRDTS-proj}

\subsection{\LTL Model Checking over Finite-Counter Systems}\label{app:LTL-CS}

We use a notion of finite-counter system that can be directly represented as a
nuXmv model. A \emph{finite-counter system} $\mathcal C$ is a tuple
\((V,\mathsf{INVAR},\mathsf{INIT},\mathsf{TRAN})\), where
\(V=\{x_1,\ldots,x_k\}\) is a finite set of variables (also called counters) ranging over
\(\mathbb{N}\), \(\mathsf{INVAR}\) and \(\mathsf{INIT}\) are linear-arithmetic formulas over \(x_1,\ldots,x_k\), and
\(\mathsf{TRAN}\) is a linear-arithmetic formula over
\(x_1,\ldots,x_k,x'_1,\ldots,x'_k\). Here, \(x'_i\) denotes the next-state copy
of \(x_i\).

A finite-counter system \(\mathcal C = (V,\mathsf{INVAR},\mathsf{INIT},\mathsf{TRAN})\)
induces a transition system \(\mathcal S_C :=(Q,Q_0,\to)\), where \(Q\) is the set of
valuations \(s:V\to\mathbb{N}\) satisfying \(\mathsf{INVAR}\),
\(Q_0\subseteq Q\) is the set of valuations satisfying \(\mathsf{INIT}\), and
\(\to\subseteq Q\times Q\) contains exactly the pairs \((s,s')\) satisfying
\(\mathsf{TRAN}\), i.e., those for which
\(\mathsf{TRAN}\) is true after substituting each \(x_i\) by \(s(x_i)\) and
each \(x'_i\) by \(s'(x_i)\). Intuitively, \(\mathsf{INVAR}\) defines
the admissible counter values, \(\mathsf{INIT}\) defines the admissible initial
counter values, and \(\mathsf{TRAN}\) relates current counter values to their possible next values.

We use a fragment of the syntax of \LTL~\cite{book:PMC}, defined according to the following grammar:
\[
\varphi ::= \mathit{true} \mid p \mid \varphi \land \varphi \mid \lnot \varphi \mid  \mathbf{G} \varphi
\]
where $\mathbf{G}$ is the "globally" operator (we do not use the "next" or "until" operators), and $p$ varies over the  \emph{atomic propositions} which are taken to be the linear-arithmetic formulas over \(V\) (rather than abstract symbols). Note that we can interpret atoms in states of the transition system $\mathcal{S_C}$ (this implicitly
defines a labeling function mapping each state to the set of
linear-arithmetic formulas true in that state, but we will not use this labeling function explicitly).

We inductively define what it means for an infinite sequence {\(\pi=\pi_0\pi_1\cdots\), with \(\pi_i:V\to\mathbb{N}\) for all \(i\),} to \emph{satisfy} an \LTL formula, as usual~\cite{book:PMC}:
\begin{itemize}
\item $\pi \models \mathit{true}$;
\item $\pi \models p$ if the atom $p$ is true under the valuation $\pi_0$;
\item $\pi \models \varphi_1 \land \varphi_2$ if $\pi \models \varphi_1$ and $\pi \models \varphi_2$;
\item $\pi \models \lnot \varphi$ if $\pi \not \models \varphi$;
\item $\pi \models \mathbf{G} \varphi$ if for all $j \geq 0$, it holds that $\pi_j \pi_{j+1} \cdots  \models \varphi$.
\end{itemize}

A \emph{path} of $\mathcal{S_C}$ is a sequence $\pi=s_0s_1s_2\cdots$ such that $s_i \in {Q}$ for all $i$, \(s_0\in Q_0\), and
\((s_i,s_{i+1}) \in \to \) for every \(i\).  A \emph{maximal} path is either an infinite path or a finite path that ends in a \emph{terminal} state $s$, i.e., for which there is no state $s' \in Q$ such that $(s,s') \in \to$. A finite maximal path  
\(s_0s_1\cdots s_n\) is said to \emph{satisfy} \(\varphi\) iff its infinite stuttering extension
\(s_0s_1\cdots s_n s_n s_n\cdots\) satisfies \(\varphi\). 

Finally, given a finite-counter system
\(\mathcal C\) and an $\LTL$ formula
\(\varphi\), we write \(\mathcal C\models\varphi\), and say that $\mathcal{C}$ \emph{satisfies} $\varphi$, iff every maximal path of the induced transition system
\(\mathcal S_{\mathcal C}\) satisfies \(\varphi\).

\subsection{Undecidability of Parameterized \textsf{HSCL}-Verification}
\label{app:undecide}

We give the detailed construction of $\mathcal T_{\mathcal C}$ and $\varphi_{\mathcal C}$ used in the proof of Theorem~\ref{thm:undec}.

\noindent The template $\mathcal T_{\mathcal C}$ has jump bound $b=1$ and is defined as follows.
\begin{itemize}
    \item \textbf{Controller locations and messages.}
    For each control state $s\in S$, there is a location $l_s$, with no broadcast message.
    For each command $\mathsf{cmd}=(s,\mathsf{op},s')\in\Delta$, there is an intermediate location $l_{\mathsf{cmd}}$.
    If $\mathsf{op}=\mathsf{inc}(c_i)$, then $l_{\mathsf{cmd}}$ broadcasts $m_{\mathsf{inc},c_i}$; if $\mathsf{op}=\mathsf{dec}(c_i)$, then it broadcasts $m_{\mathsf{dec},c_i}$; and if $\mathsf{op}=\mathsf{iszero}(c_i)$, then it broadcasts $m_{\mathsf{ack}}$.

    \item \textbf{Non-controller locations and messages.}
    There are three main non-controller locations $l_{c_1}$, $l_{c_2}$, and $l_{\mathsf{res}}$, which broadcast a corresponding ready messages $m_{\mathsf{rdy},c_1}$, $m_{\mathsf{rdy},c_2}$, and $m_{\mathsf{rdy},\mathsf{res}}$.
    For each $i\in\{1,2\}$ and each $x\in\{\mathsf{inc},\mathsf{dec}\}$, there is an intermediate location $l_{x,c_i}$, which broadcasts $m_{\mathsf{ack}}$.

    \item \textbf{Initial locations.}
    The initial locations are $l_{s_{\mathsf{init}}}$ and $l_{\mathsf{res}}$.

    \item \textbf{Type-$0$ rules.}
    For each command $\mathsf{cmd}=(s,\mathsf{op},s')\in\Delta$, the template has a rule $l_s\to l_{\mathsf{cmd}}$ with guard $m_{\mathsf{rdy},c_1}+m_{\mathsf{rdy},c_2}+m_{\mathsf{rdy},\mathsf{res}}=n-1$.
    Additionally, if $\mathsf{op}=\mathsf{iszero}(c_i)$, then the guard is conjugated with $m_{\mathsf{rdy},c_i}=0$.
    For each $i\in\{1,2\}$, there is a rule $l_{\mathsf{res}}\to l_{\mathsf{inc},c_i}$ with guard $m_{\mathsf{inc},c_i}>0$, and a rule $l_{c_i}\to l_{\mathsf{dec},c_i}$ with guard $m_{\mathsf{dec},c_i}>0$.

    \item \textbf{Type-$1$ rules.}
    All type-$1$ rules have guard $m_{\mathsf{ack}}>0$.
    For each command $\mathsf{cmd}=(s,\mathsf{op},s')\in\Delta$, there is a rule $l_{\mathsf{cmd}}\to l_{s'}$.
    For each $i\in\{1,2\}$, there are rules $l_{\mathsf{inc},c_i}\to l_{c_i}$, $l_{\mathsf{dec},c_i}\to l_{\mathsf{res}}$, $l_{c_i}\to l_{c_i}$, and $l_{\mathsf{res}}\to l_{\mathsf{res}}$.
\end{itemize}

\noindent The formula is $\varphi_{\mathcal C}:=\neg(\varphi_{\mathsf{ctrl}}\land\varphi_{\mathsf{step}}\land\varphi_{\mathsf{halt}})$, where
\[
\varphi_{\mathsf{ctrl}}:=\forall r.\ \sum_{s\in S}\kappa(l_s,r)\le 1,\qquad
\varphi_{\mathsf{step}}:=\forall r.\ \sum_{i\in\{1,2\}}\bigl(\kappa(l_{\mathsf{inc},c_i},r)+\kappa(l_{\mathsf{dec},c_i},r)\bigr)\le 1,
\]
and $\varphi_{\mathsf{halt}}:=\neg(\forall r.\ \kappa(l_{s_{\mathsf{halt}}},r)\le 0)$.
Here, $\varphi_{\mathsf{ctrl}}$ enforces that at most one process visits a controller location in each round, $\varphi_{\mathsf{step}}$ enforces that at most one process performs an increment/decrement move in each round, and $\varphi_{\mathsf{halt}}$ requires that the halting location is eventually visited.
\section{Details on Section \ref{sec:abstractions}}\label{app:abstractions}

\subsection{Receive-message abstraction}
The system $\RDTSRMA (\mathcal{T},\nu)$ removes the per-process records of received messages and eliminates $\mathsf{Receive}$ actions, while deeming $\Update(i,\rho,r)$ enabled iff there exists a receive multiset $R: \mathcal{M}\times \mathbb{N} \rightarrow \mathbb{N}$ such that (i) $0\le R\le \NState$ (pointwise), and (ii) if process $i$ had received exactly $R$, the same update would be permitted in the concrete semantics. %
Equivalently, guards are evaluated \emph{existentially}  with respect to subset of the $\NState$. %

\noindent \textbf{Projection.}
For a configuration $\cfg = \langle \PState, \NState \rangle$, we define the $\RMA$ projection $\alpha_{\RMA}(\cfg) = \langle \PState_{\RMA},\ \NState_{\RMA} \rangle,$ where \[\PState_\RMA : \llbracket 1, \nu(n) \rrbracket \rightarrow \langle \mathsf{loc}:\mathcal{L},\mathsf{rd}:\mathbb{N} \rangle\] is given, for each process $i$, by \[ \PState_{\RMA}(i) = \langle \loc \mapsto \PState(i).\loc,\ \rd \mapsto \PState(i).\rd \rangle,\] and $\NState_{\RMA} = \NState$.

The $\RMA$ semantics $\RDTSRMA(\mathcal{T},\nu) = \langle S_\nu^{\RMA},I_\nu^{\RMA},Act_\nu^{\RMA}, Tr_\nu^{\RMA}\rangle$, where 
$S_\nu^{\RMA} = \{\alpha_{\RMA}(\cfg) \mid \cfg \in S_\nu\}$, 
$I_\nu^{\RMA} = \{\alpha_{\RMA}(\cfg) \mid \cfg \in I_\nu\}$, 
the actions in $Act_\nu^{\RMA}$ are exactly the $\Update(i,\rho,r)$ actions of the concrete system, 
and the transition relation $Tr_\nu^{\RMA}$ is determined by the following enabledness and effect conditions. For $\cfg_{\RMA} = \langle \PState_\RMA, \NState_{\RMA}\rangle$ and $a = \Update(i,\rho,r)$,
\begin{align*}
\enabled_{\RMA}\big(\cfg_{\RMA}, a\big)
=\
& \PState_{\RMA}(i).(\textsf{loc, rd}) = (\rho.\frm, r) \ \ \wedge \\
& \exists R: \ 0 \le R \le \NState_{\RMA} \ \wedge \ R(\cdot, r) \models \rho.\mathsf{guard}[\mathsf{P}\leftarrow \nu].
\end{align*}
The abstract effect updates the process's location and round exactly as in the corresponding concrete effect, and modifies \(\NState_{\RMA}\) in the same way (see Appendix \ref{app:semantics-fixed}). 
{For $\pi=\cfg_0 a_1 \cfg_1 a_2 \dots$ an execution of $\RDTS^\RMA(\mathcal{T},\nu)$, we define its $\textsf{hsc}$-trace exactly as in \S~\ref{subsubsec:HSC} , i.e.,
\[
\textsf{hsc}(\pi)(\ell,r)\;=\;\bigl|\{\,k \mid a_k=\Update(i,\rho,r'),\ \rho.\mathsf{to}=\ell,\ r=r'+\rho.\mathsf{type}\,\}\bigr|.
\]}
Below we establish a forward simulation between the concrete semantics and the $\RMA$ semantics. %

\begin{lemma}[Forward simulation]\label{lem:rma-forward-simulation}
For every configuration $\cfg \in S_\nu$, if $\cfg \in I_\nu$ then $\alpha_{\RMA}(\cfg) \in I_\nu^\RMA$. Moreover if $(\cfg,a,\cfg') \in Tr_\nu$, the following hold:
\begin{enumerate}
\item If $a = \Recv(i,\langle m,r\rangle)$, then $\alpha_{\RMA}(\cfg) = \alpha_{\RMA}(\cfg')$.
\item If $a = \Update(i,\rho,r)$, then $(\alpha_{\RMA}(\cfg),a,\alpha_{\RMA}(\cfg')) \in Tr_\nu^\RMA$.
\end{enumerate}
\end{lemma}
\begin{proof}%
1. A receive transition changes only the received messages (\textsf{rcvMsg}) by process $i$; $\alpha_{\RMA}$ forgets exactly this component thus $\alpha_{\RMA}(\cfg)=\alpha_{\RMA}(\cfg')$.

\noindent 2. In the $\RDTS(\mathcal{T}, \nu)$, $a = \Update(i, \rho, r)$ enabledness requires at \textsf{cfg}:
\[\PState(i).(\textsf{loc},\textsf{rd})=(\rho.\frm,r)\text{ and }\PState(i).\mathsf{RecMsg}(\cdot, r) \models \rho.\mathsf{guard}[\mathsf{P}\leftarrow \nu].\] By equation~\ref{equ:rec-bounded-bcast},  $\PState(i).\mathsf{RecMsg}\leq \NState$. Since the $\alpha_\RMA$ projection preserves $\textsf{PState}.(\textsf{loc},\textsf{rd})$ and $\NState$, $R= \PState(i).\mathsf{RecMsg}$ (at \textsf{cfg}) witnesses abstract enabledness. The abstract effect mirrors the concrete one on the components it keeps. Therefore $(\alpha_{\RMA}(\cfg),a,\alpha_{\RMA}(\cfg')) \in Tr_\nu^\RMA$. 
\end{proof}

Next, we establish Refinement between the $\RMA$ semantics and the concrete semantics using Condition~$3$ from Def.~\ref{def:RPT}.
\begin{lemma}[Refinement]\label{lm:RMA-refinement}
For every template \(\mathcal{T}\), there exists 
a mapping %
\(\beta_{\RMA}:S_\nu^\RMA\to S_\nu\) such that:
\begin{enumerate}
    \item For every \(\cfg_\RMA\in I_\nu^\RMA\), we have that
          \(\beta_\RMA(\cfg_\RMA)\in I_\nu\).
    \item For every \((\cfg_\RMA,a,\cfg'_\RMA)\in Tr_\nu^\RMA\), 
          there exists a concrete fragment
            $\cfg_1 \xrightarrow{a_1}\cdots\xrightarrow{a_n}\cfg_{n+1}$
          with \(\beta_\RMA(\cfg_\RMA)=\cfg_1\), 
          \(\beta_\RMA(\cfg'_\RMA)=\cfg_{n+1}\), 
          where for some $1 \leq j\leq n:$ \(a_j=a\) and for all \(i\neq j\), $a_i$ is \(\Recv\) action.
\end{enumerate}
\end{lemma}

\begin{proof}[Sketch] For an $\RMA$ configuration $\cfg_\RMA = \langle \PState_\RMA, \NState_\RMA \rangle$, define the refinement mapping $\beta_{\RMA}(\cfg_{\RMA})=\cfg=\langle\PState,\NState\rangle$ with (i) $\NState = \NState_\RMA$, (ii) for process $i$ $\PState(i).(\loc, \rd) = \PState_\RMA(i).(\loc, \rd)$, and (iii)
for process \(i\) and round \(r'\):
\begin{itemize}
    \item \(\PState(i).\mathsf{rcvMsg}(\cdot,r')=\NState_\RMA(\cdot,r')\) if \(r'<\PState_\RMA(i).\mathsf{rd}\)
    \item \(\PState(i).\mathsf{rcvMsg}(\cdot,r')= 0\) if $r'>\PState_\RMA(i).\mathsf{rd}$
    \item \(\PState(i).\mathsf{rcvMsg}(m,\PState_\RMA(i).\mathsf{rd})=\NState_\RMA(m,\PState(i).\mathsf{rd})\) if \(m\in M_i(\cfg_\RMA)\) and \(0\) otherwise, where $M_i(\cfg_\RMA)$ denotes the set of message variables that occur in the guards of type-0 rules leading (via type-0 rules) to $\PState_\RMA(i).\mathsf{loc}$.
\end{itemize}

To see that $\beta_\RMA$ satisfies Part 1, note that by Def.~\ref{def:RPT} no type-0 rule targets initial control locations. Consequently, for every $\cfg_\RMA \in I_\nu^\RMA$, it follows that $M_i(\cfg_\RMA)=\emptyset$.
The key observation for Part $2$ is that if an update $a=\Update(i,\rho,r)$ is enabled in $\cfg_{\RMA}$, then by Condition~$3$ in Def.~\ref{def:RPT}, $\rho.\guard$ is monotone in the variables $M_i(\cfg_\RMA)$. Hence, assuming all such messages have been received cannot disable the update. For variables outside $M_i(\cfg_\RMA)$, receive counts are zero and thus an exact number of these can be received to execute $a$. Finally, after applying $a$ additional messages can be received to reach $\beta_{\RMA}(\cfg_\RMA')$, for $(\cfg_{\RMA}, a, \cfg_{\RMA}') \in Tr_\nu^\RMA$. Full details are given in Appendix~\ref{apdx:ref-RMA}. 
\end{proof}

\subsubsection{Proof of Proposition \ref{cor:RMA-sound-complete}}
 \begin{proof}
 \emph{Soundness.} By Lemma~\ref{lem:rma-forward-simulation}, for every execution $\pi$ of $\RDTS(\mathcal{T},\nu)$ there exits execution $\pi_{\RMA}$ of $\RDTS^\RMA(\mathcal{T},\nu)$ with the \emph{same} sequence of $\Update$ actions. Hence $\textsf{hsc}(\pi_{\RMA})=\textsf{hsc}(\pi)$. \emph{Completeness.} By Lemma~\ref{lm:RMA-refinement}, for every execution $\pi_{\RMA}$ of $\RDTS^\RMA(\mathcal{T},\nu)$ there exits execution  $\pi$ of $\RDTS(\mathcal{T},\nu)$ with the \emph{same} sequence of $\Update$ actions, interleaved with $\Recv$ actions. Since $\Recv$ actions are stuttering steps for $\textsf{hsc}$-trace, we obtain $\textsf{hsc}(\pi)=\textsf{hsc}(\pi_{\RMA})$. 
 \end{proof}

\subsection{Process-identity abstraction}

In the system \(\RDTS^\PIA(\mathcal{T},\nu)\), the \(\PState_\RMA\) component of the \(\RMA\) configuration is abstracted: instead of recording the exact location and round of each process, it records only the \emph{counts} of processes at each location and round.

\noindent \textbf{Projection.}
For a $\RMA$  configuration
$\cfg_{\RMA}=\langle \PState_{\RMA},\, \NState_{\RMA}\rangle$, we define the $\PIA$ projection $\alpha_{\PIA}(\cfg_{\RMA})= \big\langle \PState_{\PIA},\, \NState_{\PIA}\big\rangle,$ where $\PState_{\PIA} : \mathcal{L} \times \mathbb{N} \to\mathbb{N}$ is given by \[\PState_{\PIA}(\ell,r)
\;=\;
\big|\{\, i\in \llbracket 1, N\rrbracket:
\PState_{\RMA}(i).(\loc,\rd)
= (\ell,r)\,\}\big|,\] and $\NState_{\PIA} = \NState_{\RMA}$.

The $\PIA$ semantics \[\RDTS^{\PIA}(\mathcal{T},\nu) = \langle S_\nu^{\PIA}, I_\nu^{\PIA}, Act_\nu^{\PIA}, Tr_\nu^{\PIA}\rangle,\] where $S_\nu^{\PIA} = \{\, \alpha_{\PIA}(\cfg) \mid \cfg\in S_\nu^{\RMA} \,\}$, $I_\nu^{\PIA} = \{\, \alpha_{\PIA}(\cfg) \mid \cfg\in I_\nu^{\RMA} \,\},$ the action set $Act_\nu^{\PIA} = \{\, \Update(\rho,r) \mid \rho\in\mathsf{Rules},\ r\in\mathbb{N} \,\}$ (no process ids), and the transition relation $Tr_\nu^{\PIA}$ is determined by the following enabledness and effect conditions. 

For \(\cfg_{\PIA} = \alpha_{\PIA}(\cfg_{\RMA})\) 
for some \(\RMA\) configuration \(\cfg_{\RMA}\), the abstract action 
\(a=\Update(\rho,r)\) is enabled iff there exists a process \(i\) at $\cfg_{\PIA}$
\((\rho.\frm,r)\) in \(\cfg_{\RMA}\) such that \(\Update(i,\rho,r)\) is enabled 
in \(\cfg_{\RMA}\). Equivalently, for  \(\cfg_{\PIA}=\langle \PState_{\PIA},\NState_{\PIA}\rangle\),
\begin{align*}
\enabled_{\PIA}\!\big(\cfg_{\PIA}, a\big)
=\ 
& \PState_{\PIA} (\rho.\frm,\, r) > 0\ \ \wedge \\
& \exists R:\  0\le R\le \NState_{\PIA}\ \wedge\ R(\cdot, r) \models \rho.\mathsf{guard}[\mathsf{P}\leftarrow \nu].
\end{align*}
\noindent Upon executing \(a\) at \(\cfg_{\PIA}\), the $\PState_{\PIA}$ is updated to reflect move of one process from \((\rho.\frm,r)\) to \((\rho.\mathsf{to},\, r+\rho.\mathsf{type})\), and 
\(\NState_{\PIA}\) is updated exactly as in the \(\RMA\) semantics (see Appendix~\ref{app:semantics-fixed}). 
Formally, 

\smallskip
\noindent $\effect_{\PIA}(\cfg_{\PIA},\,\Update(\rho, r)) = \langle \PState'_{\PIA},\, \NState'_{\PIA} \rangle$, where:
\begin{itemize}
    \item $\PState'_{\PIA}(\ell, r') = \PState_{\PIA}(\ell, r') - 1$ if $(\ell, r') = (\rho.\frm, r)$; \\
    $\PState'_{\PIA}(\ell, r') = \PState_{\PIA}(\ell, r') + 1$ if $(\ell, r') = (\rho.\mathsf{to},\, r+\rho.\mathsf{type})$;\\
    and $\PState'_{\PIA}(\ell, r') = \PState_{\PIA}(\ell, r')$ otherwise. 
    \item $\NState'_{\PIA}(m', r') = \NState_{\PIA}(m', r') + 1$ if $(m', r') = (\mathsf{Bcast}(\rho.\mathsf{to}),\, r+\rho.\mathsf{type})$; and $\NState'_{\PIA}(m', r') = \NState_{\PIA}(m', r')$ otherwise.
\end{itemize}
{For $\pi=\cfg_0 a_1 \cfg_1 a_2 \dots$ an execution of $\RDTS^\PIA(\mathcal{T},\nu)$, define its $\textsf{hsc}$-trace analogously to as defined in \S~\ref{subsubsec:HSC}, except that here updates no longer carry process identifiers:  
\[
\textsf{hsc}(\pi)(\ell,r)\;=\;\bigl|\{\,k \mid a_k=\Update(\rho,r'),\ \rho.\mathsf{to}=\ell,\ r=r'+\rho.\mathsf{type}\,\}\bigr|.
\]}\label{def:HSC-PIA}

Define the relabeling $\delta : Act_\nu^{\RMA} \to Act_\nu^{\PIA}$ by $\delta(\Update(i,\rho,r)) = \Update(\rho,r)$. 
We now establish an action-based bisimulation, modulo $\delta$, between the $\RMA$ and $\PIA$ semantics. 

\begin{lemma}[Action-based bisimulation]\label{lem:PIA-bisim} 
\begin{enumerate}
 \item  $\forall\ \cfg \in I_\nu^{\RMA}:$  $\alpha_{\PIA}(\cfg) \in I_\nu^{\PIA}$ and
  $ \forall\ \cfg_{\PIA} \in I_\nu^{\PIA}:$  $\exists\,\cfg \in I_\nu^{\RMA}$ with $\alpha_{\PIA}(\cfg)=\cfg_{\PIA}$.
\item For all $\cfg \in S_\nu^\RMA$ it holds:
\begin{itemize}
    \item[$\Rightarrow$] If \((\cfg,a,\cfg')\in Tr_\nu^{\RMA}\), then \((\alpha_{\PIA}(\cfg),\delta(a),\alpha_{\PIA}(\cfg')) \in Tr_\nu^{\PIA}\).
    \item[$\Leftarrow$] If $(\alpha_{\PIA}(\cfg),a',\cfg'_{\PIA})\in Tr_\nu^{\PIA}$, then $\exists\ \cfg'\in S_\nu^{\RMA}, a \in Act_\nu^\RMA:$ $(\cfg,a,\cfg')\in Tr_\nu^{\RMA}$, $\alpha_{\PIA}(\cfg')=\cfg'_{\PIA}$, and $\delta(a) = a'$.
\end{itemize}
\end{enumerate}
\end{lemma}
\begin{proof}
By definition of \(I_\nu^{\PIA}=\{\alpha_{\PIA}(\cfg)\mid \cfg\in I_\nu^{\RMA}\}\), thus we have part 1 of the lemma.

\noindent  ($\Rightarrow$)
Let $(\cfg,a,\cfg')\in Tr_\nu^{\RMA}$ with $a=\Update(i,\rho,r)$.
Enabledness of $a$ in $\RMA$ semantics requires that in $\cfg$:
\begin{itemize}
    \item $\PState_{\RMA}(i).(\loc,\rd)=(\rho.\frm,r)$, hence $\PState_{\PIA}(\rho.\frm,r)>0$ in $\alpha_{\PIA}(\cfg)$.
    \item there exists $R$ with $0\le R\le \NState_{\RMA}$ and $R(\cdot, r)\models \rho.\guard[\mathsf P\leftarrow\nu]$.
\end{itemize}
Since $\NState_{\PIA}=\NState_{\RMA}$ under $\alpha_{\PIA}$, the same $R$ witnesses enabledness of $a'=\delta(a)=\Update(\rho,r)$ in $\alpha_{\PIA}(\cfg)$.
For the effect, $a$ moves one process from $(\rho.\frm,r)$ to $(\rho.\mathsf{to},r+\rho.\type)$ and (if applicable) broadcasts to $\NState$. Thus $\alpha_{\PIA}(\cfg')$ decreases $\PState_{\PIA}(\rho.\frm,r)$ by $1$, increases $\PState_{\PIA}(\rho.\mathsf{to},r+\rho.\type)$ by $1$, and updates $\NState$ exactly as in $\effect_{\PIA}$.
Hence $(\alpha_{\PIA}(\cfg),a',\alpha_{\PIA}(\cfg'))\in Tr_\nu^{\PIA}$.

\noindent ($\Leftarrow$)
Let $(\alpha_{\PIA}(\cfg),a',\cfg'_{\PIA})\in Tr_\nu^{\PIA}$ with $a'=\Update(\rho,r)$.
Enabledness of $a'$ in $\PIA$ semantics requires that in $\alpha_{\PIA}(\cfg)$:
\begin{itemize}
    \item $\PState_{\PIA}(\rho.\frm,r)>0$, hence there exists $i$ with $\PState_{\RMA}(i).(\loc,\rd)=(\rho.\frm,r)$ in $\cfg$.
    \item there exists $R$ with $0\le R\le \NState_{\PIA}$ and $R(\cdot, r)\models \rho.\guard[\mathsf P\leftarrow \nu]$.
\end{itemize}
Since $\NState_{\PIA}=\NState_{\RMA}$ under $\alpha_{\PIA}$, the same $R$ witnesses enabledness of $a=\Update(i,\rho,r)$ in $\cfg$.
Let \(\cfg'\) be the \(\RMA\) configuration upon executing \(a=\Update(i,\rho,r)\) from $\cfg$. 
Executing $a$ moves one process from $(\rho.\frm,r)$ to $(\rho.\mathsf{to},r+\rho.\type)$ and (if applicable) broadcasts to $\NState$. Thus $\alpha_{\PIA}(\cfg')$ decreases $\PState_{\PIA}(\rho.\frm,r)$ by $1$, increases $\PState_{\PIA}(\rho.\mathsf{to},r+\rho.\type)$ by $1$, and updates $\NState$ exactly as in $\effect_{\PIA}$.  Hence $(\cfg,a,\cfg')\in Tr_\nu^{\RMA}$ with $\delta(a)=a'$ and \(\alpha_{\PIA}(\cfg')=\cfg'_{\PIA}\).
\end{proof}

\subsubsection{Proof of Proposition \ref{lem:PIA-completeness}}
 \begin{proof}%
 Let $\pi = \cfg_1 a_1 \cfg_2 a_2 \cdots$ be an execution of $\RDTS^\RMA(\mathcal{T}, \nu)$.  
 Define its projection 
 \[\Delta(\pi) = \alpha_\RMA(\cfg_1)\ \delta(a_1)\ \alpha_\RMA(\cfg_2)\ \delta(a_2)\ \cdots.\]  
 By the direction ($\Rightarrow$) of Lem.~\ref{lem:PIA-bisim},  
 $\Delta(\pi)$ is an execution of $\RDTS^\PIA(\mathcal{T}, \nu)$. Moreover $\textsf{hsc}(\Delta(\pi)) = \textsf{hsc}(\pi).$
 Thus, soundness follows.
 For completeness, apply direction ($\Leftarrow$) of Lem.~\ref{lem:PIA-bisim}:  
 for every execution $\pi_\PIA$ of $\RDTS^\PIA(\mathcal{T}, \nu)$ there exists an execution $\pi$ of $\RDTS^\RMA(\mathcal{T}, \nu)$ such that $\Delta(\pi)=\pi_\PIA$. 
\end{proof}

\subsection{Semi-synchronous restriction}
\label{app:SSR}
\subsubsection{\SSR}
The system  $\RDTS^{\SSR}(\mathcal{T},\nu) = \langle S_\nu^{\SSR}, I_\nu^{\SSR}, Act_\nu^{\SSR}, Tr_\nu^{\SSR} \rangle,$  where $S_\nu^{\SSR} = S_\nu^{\PIA}$, $I_\nu^{\SSR} = I_\nu^{\PIA}$, $Act_\nu^{\SSR} = Act_\nu^{\PIA}$, and \[Tr_\nu^{\SSR} = \{(\cfg,a,\cfg')\in Tr_\nu^{\PIA}\ |\ \tgt(a) \geq \rmax(\cfg)\}.\] 
Equivalently, enabledness and effect are given by \[\enabled_{\SSR}(\cfg,a) = \enabled_{\PIA}(\cfg,a) \wedge \tgt(a) \geq \rmax(\cfg)\] and $\effect_{\SSR}(\cfg,a) = \effect_{\PIA}(\cfg,a)$. 

{The $\textsf{hsc}$-trace of executions of $\RDTS^{\SSR}(\mathcal{T},\nu)$ is defined identically to that of executions of $\RDTS^{\PIA}(\mathcal{T},\nu)$ (see Appendix~\ref{def:HSC-PIA}). Since the action sets of the two systems coincide, this is well defined.}

Recall that every rule $\rho \in \Rules$ satisfies $\rho.\type \leq b$, where $b \in \mathbb{N}$ is the \emph{jump bound} of the template $\mathcal{T}$. Hence, in a configuration $\cfg$, no update $\Update(\rho,r)$ with $r < \rmax(\cfg) - b$ is enabled under $\SSR$. Consequently, in configuration $\cfg$ all counters below $\rmax(\cfg) - b$ can be discarded, while counters above $\rmax(\cfg)$ are zero by definition. Thus the relevant window of any $\SSR$ configuration $\cfg$ is $\llbracket \rmax(\cfg) - b, \rmax(\cfg)\rrbracket$, of fixed size $b+1$, uniformly bounded across all configurations. 
For an update action $a=\Update(\rho,r)$, let $\tgt(a)=r+\rho.\type$. 

An execution \(\pi=\cfg_1\,a_1\,\cfg_2\,a_2\cdots\) of the \(\RDTS^\PIA(\mathcal{T}, \nu)\) is called \emph{semi-synchronous} if its target sequence \(\tgt(a_1)\tgt(a_2)\dots\) is non-decreasing. This notion precisely characterizes executions of $\RDTS^{\SSR}(\mathcal{T},\nu)$, as formalized and proved in following lemma. %

\begin{lemma}\label{prop:SSRexec=SSexec}
An execution $\pi=\cfg_1a_1\cfg_2a_2\dots$ of $\RDTS^\PIA(\mathcal{T}, \nu)$ is semi synchronous iff it is also an execution of $\RDTS^\SSR(\mathcal{T}, \nu)$.
\end{lemma}

\begin{proof} To prove the lemma we first prove following claim:
\begin{claim}
    For any $(\cfg,a,\cfg')\in Tr_\nu^\PIA$,  $\rmax(\cfg')=\max\bigl(\rmax(\cfg),\,\tgt(a)\bigr).$
\end{claim}
\noindent \emph{Proof of claim.} By construction of $\PIA$ an action $a=\Update(\rho,r)$ modifies only counters at rounds $r$ and $\tgt(a)>r$. If $\tgt(a)<\rmax(\cfg)$, then $\rmax(\cfg')=\rmax(\cfg)$. If $\tgt(a)\ge \rmax(\cfg)$, the transition increments the $\rho.\mathsf{to}$ location counter at round $\tgt(a)$, making it positive, so $\rmax(\cfg')=\tgt(a)$. In both cases the claim holds. 

Now suppose $\pi$ is semi-synchronous execution of $\RDTS^\PIA(\mathcal{T}, \nu)$, i.e., $\tgt(a_1)\le\tgt(a_2)\le\cdots$. Since every $\cfg\in I_\nu^\PIA$ satisfies $\rmax(\cfg)=1$, we have $\tgt(a_1)\ge \rmax(\cfg_1)$. By the claim $\rmax(\cfg_2)=\tgt(a_1)$. Inductively, if $\rmax(\cfg_i)=\tgt(a_{i-1})$ then $\tgt(a_i)\ge\tgt(a_{i-1})=\rmax(\cfg_i)$, so the $\SSR$ enabledness restriction holds at step $i$. Thus $\pi$ is an execution of $\RDTS^\SSR(\mathcal{T}, \nu)$.  

Conversely, suppose $\pi$ is an execution of $\RDTS^\SSR(\mathcal{T}, \nu)$. Then by definition $\tgt(a_i)\ge \rmax(\cfg_i)$ for all $i$. The above claim gives $\rmax(\cfg_{i+1})=\tgt(a_i)$, which implies $\tgt(a_{i+1})\ge \rmax(\cfg_{i+1})=\tgt(a_i)$. Therefore $\tgt(a_1)\le\tgt(a_2)\le\cdots$, i.e., the target sequence is non-decreasing, so $\pi$ is semi-synchronous.   
\end{proof}

To prove Lemma \ref{lm:SSR-soundness} we use the following claim whose proof follows from the definition of $Tr_\nu^\PIA$ and can be found in Appendix \ref{apdx:soundness-SSR}.

\begin{claim}[Commutativity]%
Let $a=\Update(\rho,r)$ and $a'=\Update(\rho',r')$, and suppose we have $(\cfg, a, \cfg^*), (\cfg^*, a', \cfg') \in Tr_\nu^\PIA$. If $\tgt(a)\neq r'$, then there exists $\cfg^\dagger$ with $(\cfg, a', \cfg^\dagger), (\cfg^\dagger, a, \cfg') \in Tr_\nu^\PIA$.
\end{claim}

\subsubsection{Proof of Lemma \ref{lm:SSR-soundness}.}

\begin{proof}
Let $\pi=\cfg_1 a_1 \cfg_2 a_2 \cdots$ be an execution of $\RDTS^{\PIA}(\mathcal{T},\nu)$, where each $a_i=\Update(\rho_i,r_i)$ and $\tgt(a_i)=r_i+\rho_i.\type$. 
We transform $\pi$ into a semi-synchronous execution by repeatedly commuting out-of-order updates until the sequence of target rounds becomes non-decreasing.

\noindent\textit{Inductive reordering.}
Write $\pi_{-1}=\pi$. For $r=0,1,\dots$ define $\pi_r$ from $\pi_{r-1}$ as follows.
Scan $\pi_{r-1}$ left-to-right and identify (in increasing order of their indices) all updates whose target is $r$. There are finitely many such updates by finiteness of processes and acyclicity of type $0$ updates.
For each such update, repeatedly apply the Commutativity Claim to swap it left over any immediately preceding update whose target is $>r$, stopping once the preceding update (if any) has target $\le r$.
Perform these swaps in the order the target-$r$ updates are encountered so their relative order is preserved.
Denote the resulting execution by $\pi_r$.

\noindent\textit{Invariants.} For every $r\ge 0$ there exists an index $m_r$ such that:
(i) the prefix execution $\pi_r[1..m_r]$ consists \emph{exactly} of all updates with target $\le r$, arranged in non-decreasing order of their targets, and for each fixed target the original relative order is preserved;
(ii) every update in the suffix $\pi_r[m_r+1..\ ]$ has target $>r$;
(iii) $\pi_r$ is a valid $\PIA$ execution starting from $\cfg_1$; and
(iv) $\pi_r$ has the same multiset of update actions as $\pi$.
Items (iii)–(iv) hold because each step uses only swaps justified by the Commutativity Claim; items (i)–(ii) hold by construction.

\noindent\textit{Limit execution.}
For $n\in\mathbb{N}$, let $r_0$ be the maximum target round among the first $n$ updates of $\pi$. Then for all $r>r_0$, the executions $\pi_r$ and $\pi_{r_0}$ agree on the prefix of length $n$. thus the prefix of length $n$ stabilizes. Consequently, the pointwise limit $\pi'=\lim_{r\to\infty}\pi_r$ exists and is a valid execution, obtained by taking, for each $n$, the stabilized prefix of length $n$. 
\end{proof}

\subsubsection{Proof of Proposition \ref{lem:SSR-completeness}}
\begin{proof}
By Lemma~\ref{lm:SSR-soundness}, for every execution $\pi$ of $\RDTS^{\PIA}(\mathcal{T},\nu)$, there exists semi-synchronous execution $\pi'$ of $\RDTS^{\PIA}(\mathcal{T},\nu)$ such that $\textsf{hsc}(\pi') = \textsf{hsc}(\pi)$. Furthermore, by Lemma \ref{prop:SSRexec=SSexec}, the executions of $\RDTS^{\SSR}(\mathcal{T},\nu)$ are exactly the semi-synchronous executions of $\RDTS^{\PIA}(\mathcal{T},\nu)$. \\Thus the set of $\textsf{hsc}$-traces of $\RDTS^{\PIA}(\mathcal{T},\nu)$ and $\RDTS^{\SSR}(\mathcal{T},\nu)$ is identical. 
\end{proof}
\subsubsection{Strong \textsf{SSR}}

{The $\SSR$ restriction is sufficient to obtain a reduced finite-counter semantics. However, we can further decrease non-determinism on interleaving by introducing a stronger restriction while preserving soundness.}

{Let $\theta$ be a total order on the set of jump rules $\{\rho \in \Rules\ |\ \rho.\type > 0\}$. An execution $\pi = \cfg_0, a_1, \cfg_1, a_2 \dots$ of $\RDTS^{\PIA}(\mathcal{T},\nu)$ with $a_i = \Update(\rho_i, r_i)$ is called \emph{$\theta$-strong semi-synchronous} if:
(i) it is semi-synchronous, i.e., for all $i<j$, $\tgt(a_i) \leq \tgt(a_j)$, and
(ii) whenever $\tgt(a_i) = \tgt(a_j)$ and $\rho_j.\type > 0$, we have $\rho_i.\type > 0$ and $\rho_i \leq_\theta \rho_j$.
Clearly, every $\theta$-strong semi-synchronous execution is semi-synchronous. The converse is also true:}

\begin{lemma}[Soundness of strong $\SSR$]
Let $\theta$ be a total order on jump rules. For every semi-synchronous execution $\pi$ of $\RDTS^{\PIA}(\mathcal{T},\nu)$ there exists a $\theta$-strong semi-synchronous execution $\pi'$ of $\RDTS^{\PIA}(\mathcal{T},\nu)$ that starts from the same initial configuration as $\pi$ and contains the same multiset of update actions.
\end{lemma}
\begin{proof}
The proof follows from the Commutativity claim. In a semi-synchronous execution, within the fragment whose updates targets round $r$, a local update targeting round $r$ can always be commuted past a jump update targeting $r$, so all jumps can be placed before locals. Furthermore, jump updates targeting the same round commute among themselves. Thus, by reordering them according to $\theta$, we obtain the desired execution. 
\end{proof}
{The corresponding strong $\SSR$ semantics extends the $\SSR$ semantics. Informally, configurations additionally record the last executed rule, and enabledness is further restricted as follows: (i) if the last rule was local, then no jump to the current frontier round is permitted; (ii) if the last rule was a jump, then no jump to the current frontier round with a smaller (w.r.t. $\theta$) rule is permitted.
Strong $\SSR$ substantially reduces nondeterminism:
(i) it enforces that all jumps to a round $r$ occur before any locals at $r$, and
(ii) it quotients away permutations of jump updates targeting the same round by fixing their order according to $\theta$.
Since soundness holds for any choice of $\theta$, the strong $\SSR$ semantics can be viewed as if all processes \emph{synchronously jump} to round $r$, followed by an \emph{asynchronous execution of local updates} at $r$, before moving on to the next round.}

\subsection{Bounded-window abstraction}

 Recall that in $\SSR$ the relevant window of a configuration $\cfg$ is $\llbracket \rmax(\cfg)-b,\ \rmax(\cfg)\rrbracket$, of fixed size $b{+}1$. The bounded-window abstraction ($\BWA$) drops all counters outside this window and encodes the remainder as a sliding window of width $b{+}1$ anchored at the frontier $\rmax(\cfg)$; that is, it stores $\rmax(\cfg)$ together with the counters for rounds $\rmax(\cfg)-b\leq r \leq \rmax(\cfg)$.

 \noindent \textbf{Projection.} 
For a $\SSR$  configuration
$\cfg_{\SSR}=\langle \PState_{\SSR},\, \NState_{\SSR}\rangle$, we define the $\SSR$ projection \[\alpha_{\BWA}(\cfg_{\SSR})= \big\langle \PState_{\BWA},\, \NState_{\BWA},\, \rmax\big\rangle,\]  where $\PState_{\BWA} : \mathcal{L} \times \llbracket 0, b\rrbracket \to\mathbb{N}$ is given by $\PState_{\BWA}(\ell, d) = \PState_{\SSR}(\ell, \rmax(\cfg)-d),$ $\NState_{\BWA} : \mathcal{M} \times \llbracket 0, b\rrbracket \to\mathbb{N}$ is given by $\NState_{\BWA}(m, d) = \NState_{\SSR}(m, \rmax(\cfg)-d),$ and $\rmax \in \mathbb{N}$ is given by $\rmax = \rmax(\cfg_\SSR)$.
The $\BWA$ semantics \[ \RDTS^\BWA(\mathcal{T},\nu) = \langle S_\nu^{\BWA}, I_\nu^{\BWA}, Act_\nu^{\BWA}, Tr_\nu^{\BWA}\rangle, \] where $S_\nu^{\BWA} = \{\alpha_{\BWA}(\cfg) \mid \cfg \in S_\nu^\SSR\}$, $I_\nu^{\BWA} = \{\alpha_{\BWA}(\cfg) \mid \cfg \in I_\nu^\SSR\}$, and $Act_\nu^{\BWA} = Act_\nu^{\SSR}$.
The transition relation is the existential lift of $\SSR$ semantics: $(\cfg_{\BWA},a,\cfg'_{\BWA})\in Tr_\nu^{\BWA}$ iff there exist $\cfg_{\SSR},\cfg'_{\SSR}\in S_\nu^{\SSR}$ with
\[\alpha_{\BWA}(\cfg_{\SSR})=\cfg_{\BWA},\ \alpha_{\BWA}(\cfg'_{\SSR})=\cfg'_{\BWA}, \text{ and }(\cfg_{\SSR},a,\cfg'_{\SSR})\in Tr_\nu^{\SSR}.\]
{The $\textsf{hsc}$-trace of executions of $\RDTS^{\BWA}(\mathcal{T},\nu)$ is defined identically to that of executions of $\RDTS^{\PIA}(\mathcal{T},\nu)$ (see Appendix~\ref{def:HSC-PIA}). Since the action sets of the two systems coincide, this is well defined.}

Before establishing action based bisimulation between $\BWA$ and $\SSR$ semantics, we establish following property of projection $\alpha_\BWA:$
\begin{claim}%
Let $\cfg,\hat\cfg\in S_\nu^\SSR$ and $a\in Act_\nu^\SSR$.
If $\alpha_\BWA(\cfg)=\alpha_\BWA(\hat\cfg)$, then in $\SSR$ semantics:
\begin{enumerate}
\item $a$ is enabled at $\cfg$  iff it is enabled at $\hat\cfg$.
\item  if $a$ is enabled, then $\alpha_\BWA\!\big(\effect_\SSR(\cfg,a)\big)\;=\;\alpha_\BWA\!\big(\effect_\SSR(\hat\cfg,a)\big).$

\end{enumerate}
\end{claim}
\begin{proof}
Assume $\alpha_\BWA(\cfg)=\alpha_\BWA(\hat\cfg)$. Then $\rmax(\cfg)=\rmax(\hat\cfg)$ and, for all $r \geq \rmax(\cfg) -b$, the round $r$ counters,
$\PState(\ell,r)$ and $\NState(m,r)$ coincide in $\cfg$ and $\hat\cfg$.

\noindent 1.
For any $a=\Update(\rho,r)$, $a$ is enabled at $\cfg$ implies $r \geq \rmax(\cfg) -b$, Since the enabledness of $a$ depends only on the round $r$ counters, $a$ is enabled in $\cfg$ iff it is enabled in $\hat\cfg$. 

\noindent 2. When enabled, $a$ updates only round $r$ and round $\tgt(a)$ counters and since $\tgt(a) \geq r$ for all rounds $r' \geq \rmax(\cfg) -b$ the round $r'$ counters coincide in $\cfg' = \effect_\SSR(\cfg,a)$ and $\hat{\cfg}' = \effect_\SSR(\hat\cfg,a)$. Moreover \[\rmax(\cfg'))=\max\{\rmax(\cfg),\,\tgt(a)\} = \max\{\rmax(\hat \cfg),\,\tgt(a)\} = \rmax(\hat{\cfg}').\] Thus $\cfg'$ and $\hat{\cfg}'$ agree on $\rmax$ and the round $r'$ counters for $r' \geq \rmax(\cfg')) - b \geq \rmax(\cfg) -b$. Therefore the updated window counters match as well, showing
$\alpha_\BWA(\cfg')=\alpha_\BWA(\hat{\cfg}')$. 

\end{proof}

\begin{lemma}[Action-based bisimulation]\label{lem:BWA-bisim} It holds that:
\begin{enumerate}
 \item  $\forall\ \cfg \in I_\nu^{\SSR}:$  $\alpha_{\BWA}(\cfg) \in I_\nu^{\BWA}$ and
  $ \forall\ \cfg_{\BWA} \in I_\nu^{\BWA}:$  $\exists\,\cfg \in I_\nu^{\SSR}$ with $\alpha_{\BWA}(\cfg)=\cfg_{\BWA}$.
\item For all $\cfg \in S_\nu^\SSR$ it holds:
\begin{itemize}
    \item[$\Rightarrow$] If \((\cfg,a,\cfg')\in Tr_\nu^{\SSR}\), then \((\alpha_{\BWA}(\cfg),a,\alpha_{\BWA}(\cfg')) \in Tr_\nu^{\BWA}\).
    \item[$\Leftarrow$] If $(\alpha_{\BWA}(\cfg),a,\cfg'_{\BWA})\in Tr_\nu^{\BWA}$, then $\exists\ \cfg'\in S_\nu^{\SSR}:$ $(\cfg,a,\cfg')\in Tr_\nu^{\SSR}$ and $\alpha_{\BWA}(\cfg')=\cfg'_{\BWA}$.
\end{itemize}
\end{enumerate}
\end{lemma}
\begin{proof}
(1) holds immediately since $I_\nu^{\BWA}=\{\alpha_{\BWA}(\cfg)\mid \cfg\in I_\nu^{\SSR}\}$.  
Direction ($\Rightarrow$) holds by the definition of $Tr_\nu^{\BWA}$ as an existential lift of $Tr_\nu^{\SSR}$.  

For direction ($\Leftarrow$), let $\cfg\in S_\nu^\SSR$ with $\alpha_{\BWA}(\cfg)=\cfg_\BWA$ and assume $(\cfg_\BWA,a,\cfg'_\BWA)\in Tr_\nu^\BWA$.
Then by definition of $Tr_\nu^\BWA$, there exist $\cfg^\dagger,\cfg^{\dagger\prime}\in S_\nu^\SSR$ with $(\cfg^\dagger,a,\cfg^{\dagger\prime})\in Tr_\nu^\SSR$, $\alpha_\BWA(\cfg^\dagger)=\cfg_\BWA$, and $\alpha_\BWA(\cfg^{\dagger\prime})=\cfg'_\BWA$.  
Since $\cfg$ and $\cfg^\dagger$ have the same $\BWA$ projection, above Claim implies that $a$ is also enabled at $\cfg$ and that \[\alpha_\BWA(\effect_\SSR(\cfg,a))=\alpha_\BWA(\effect_\SSR(\cfg^\dagger,a))=\cfg'_\BWA.\]  
Taking $\cfg'=\effect_\SSR(\cfg,a)$ proves the direction. 
\end{proof}

\subsubsection{Proof of Proposition \ref{lem:BWA-completeness}}
\begin{proof}
 By Lem. \ref{lem:BWA-bisim} $\alpha_{\BWA}$ is an action-based bisimulation between $\RDTS^{\SSR}(\mathcal{T},\nu)$ and $\RDTS^{\BWA}(\mathcal{T},\nu)$.  Thus, for every execution $\pi$ of $\RDTS^{\SSR}(\mathcal{T},\nu)$ there exists an execution $\pi_\BWA$ of $\RDTS^{\BWA}(\mathcal{T},\nu)$ with the same sequence of action labels, and hence $\textsf{hsc}(\pi_\BWA) = \textsf{hsc}(\pi)$.  
The converse direction follows symmetrically. 
\end{proof}

\subsection{History-record extension}

Formally, $\RDTS^\HRE(\mathcal{T},\nu,\varphi)
= \langle S_\nu^{\HRE}, I_\nu^{\HRE}, Act_\nu^{\HRE}, Tr_\nu^{\HRE}\rangle,$ where $S_\nu^{\HRE}$ is the set of $\HRE$ configurations, $I_\nu^{\HRE}$ the set of initial $\HRE$ configurations, $Act_\nu^{\HRE}=Act_\nu^{\BWA}$, and $Tr_\nu^{\HRE}$, the transition relation obtained by extending $Tr_\nu^{\BWA}$ with the corresponding updates to the history record.

\noindent \textbf{$\HRE$ Configurations.}
An $\HRE$ configuration is of the form \[\cfg_\HRE=\langle \cfg_\BWA,\local,\cumul\rangle\] with $\cfg_\BWA\in S_\nu^\BWA$ and  $\local,\cumul:\mathcal{X}\to\mathbb{N}$. For $\ell\in\mathcal{X}$, $\local(\ell)$ counts visits to $\ell$ in the frontier round $\rmax$ of $\cfg_\BWA$, and $\cumul(\ell)$ counts total visits to $\ell$ up to~$\rmax$. We refer to counters $\local$ and $\cumul$ as history record counters.

For a formula $\varphi$, define $t_{\max}(\varphi, \nu)$ as the largest numeric value attained by any threshold in $\varphi$ under $\nu$. Formally, let $T(\varphi) \subseteq \textsf{LinTrm}(\mathcal{Y})$ denote the set of thresholds occurring in $\varphi$ (for example, in Figure~\ref{fig:prop-expressivity}-Left, we have $T(\texttt{A}) = \{0\}$ and $T(\texttt{T}) = \{\mathsf{N_c}\}$), then 
\[t_{\max}(\varphi, \nu) = \max \{\, t[\mathcal{Y} \leftarrow \nu] \mid t \in T(\varphi)\,\}.\]
In $\HRE$ configurations we require that for every $\ell \in \mathcal{X}$, both $\local(\ell)$ and $\cumul(\ell)$ are bounded above by $B_\nu^\varphi = 1 + \max \{0,\, t_{\max}(\varphi, \nu)\}.$
This bounding is sufficient because every atomic constraint in $\varphi$ has the form $\sum_{x\in\mathcal{X}} c_x A_x \le t$ where $c_x \in \mathbb{N}$ and $A_x$ denote either a local or cumulative visit count. If $c_x=0$, the value of $A_x$ is irrelevant. If $c_x>0$ and $A_x$ takes value $B_\nu^\varphi$, then the atom is permanently false regardless of further increases.
Thus exact values of $A_x$ beyond $B_\nu^\varphi$ never affect the truth of~$\varphi$.

An $\HRE$ configuration $\cfg_\HRE=\langle \cfg_\BWA,\local, \cumul\rangle$ is initial iff $\cfg_\BWA\in I_\nu^{\BWA}$ and, the history record counters are initialized to match the number of processes at location $\ell$ in round~0, respecting the bound of $B_\nu^\varphi$.
Formally, for every $\ell\in\mathcal{X}$, 
\[\local(\ell)=\cumul(\ell)=\min\bigl(\PState_\BWA(\ell,0), B_\nu^\varphi\bigr),\]
where $\PState_\BWA$ is the process-state component of $\cfg_\BWA$.

\noindent \textbf{$\HRE$ transitions.}
For $\HRE$ configurations $\cfg_\HRE=\langle \cfg_\BWA,\local, \cumul\rangle$ and $\cfg_\HRE'=\langle \cfg_\BWA',\local', \cumul'\rangle$, and an action $a=\Update(\rho,r)$, we have\\ $(\cfg_\HRE,a,\cfg_\HRE)\in Tr_\nu^{\HRE}$ iff $(\cfg_\BWA,a,\cfg_\BWA')\in Tr_\nu^{\BWA}$ and the history record counters are related as follows.
Let $\rmax$, $\rmax'$ be the frontier rounds of $\cfg_\BWA$ and $\cfg_\BWA'$ respectively. Then, for each $\ell \in \mathcal{X}$:
\begin{itemize}
    \item The cumulative record $\cumul$ is incremented at location $\rho.\mathsf{to}$ subject to the saturation at $B_\nu^\varphi$; other locations are unchanged. Formally,  $\mathsf{cumul}'(\ell)=\min(\mathsf{cumul}(\ell)+1, B_\nu^\varphi)$ if $\ell=\rho.\mathsf{to}$ and $\mathsf{cumul}'(\ell)=\mathsf{cumul}(\ell)$ otherwise.%
    \item Recall that the record $\local$ counts visits in the frontier round. 
    \begin{itemize}
        \item If the frontier-round does not advance ($\rmax'=\rmax$), then only the target location increments, subject to the saturation. Formally, $\local'(\ell)=$ $ \min(\local(\ell)+1, B_\nu^\varphi)$ if  $\ell=\rho.\mathsf{to}$ and $\local'(\ell)= \local(\ell)$ otherwise.
        \item If the frontier-round advances ($\rmax'>\rmax$), all local counters reset to $0$ except at the new target, which is set to $1$. Formally, $\local'(\ell)= 1$ if  $\ell=\rho.\mathsf{to}$ and $\local'(\ell)= 0$ otherwise.
        \item Note that frontier-round is never decremented in $\BWA$ semantics. 
    \end{itemize}
\end{itemize}   

\subsubsection{Proof of Proposition \ref{lem:LTL-transaltion}}

\begin{proof}%
Write $\gamma(\langle \cfg_\BWA,\local, \cumul \rangle) = \cfg_\BWA$.  
For every $\HRE$ execution $\pi=\cfg_0 \xrightarrow{a_0}\cfg_1 \xrightarrow{a_1}\cdots$, its $\gamma$ projection
\[
\gamma(\pi) = \gamma(\cfg_0) \xrightarrow{a_0}\gamma(\cfg_1) \xrightarrow{a_1}\cdots
\]
is a $\BWA$ execution.  
Conversely, for every $\BWA$ execution $\pi=\cfg_0 \xrightarrow{a_0}\cfg_1 \xrightarrow{a_1}\cdots$, there exists $\HRE$ execution
\[
\pi' = \langle \cfg_0,{\local}_0, \cumul_0\rangle \xrightarrow{a_0} \langle \cfg_1,{\local}_1, \cumul_1\rangle \xrightarrow{a_1} \cdots,
\]
obtained by initializing history record counters as specified in $I_\nu^{\HRE}$ and executing the same actions and updating the history record counters as specified by $Tr_\nu^{\HRE}$.

We now prove the proposition for base cases $\varphi = \forall_r\,\alpha_r$ and $\varphi = \beta$. Since the translation is preserved for $\neg$ and $\wedge$, the inductive cases follow immediately.

\begin{itemize}
    \item \textbf{Case $\varphi = \forall_r\,\alpha_r$:}  
     Assume $\alpha_r=\sum_{\ell} c_\ell\cdot\kappa(\ell,r)\le t$.  
    To show $\pi \models \varphi$ iff $\pi' \models \LTL(\varphi)$.
    \begin{itemize}
        \item ($\Rightarrow$) Suppose $\pi \models \varphi$ but $\pi' \not\models \LTL(\varphi)$. Then there exists some position $i$ in $\pi'$ with
    \[
    \sum_{\ell \in \mathcal{X}} c_\ell \cdot {\local}_i(\ell) > t[\mathcal{Y} \leftarrow \nu].
    \]
    Since $\local$ counts process visits to round $\rmax(\cfg_i)$, this contradicts $\pi \models \forall_r\,\alpha_r$, as $\alpha_r$ fails for $r=\rmax(\cfg_i)$.
    \item ($\Leftarrow$) Conversely, suppose $\pi' \models \LTL(\varphi)$ but $\pi \not\models \varphi$.  
    Then there exists a round $r_0$ with
    \[
    \sum_{\ell} c_\ell\cdot\textsf{hsc}(\pi)(\ell,r_0) > t[\mathcal{Y}\leftarrow \nu].
    \]
    Since the number of processes is finite and each performs only a bounded number of updates per round, there exists a last position $i$ in $\pi$ where $\rmax(\cfg_i)=r_0$ (and if $\cfg_{i+1}$ exists, then $\rmax(\cfg_{i+1})>r_0$).
    At this point,
    \[
    \sum_{\ell \in \mathcal{X}} c_\ell \cdot {\local}_i(\ell) = \sum_{\ell \in \mathcal{X}} c_\ell\cdot\textsf{hsc}(\pi)(\ell,r_0) > t[\mathcal{Y} \leftarrow \nu],
    \]
    contradicting $\pi' \models \LTL(\varphi)$.  
    \end{itemize}
    \item \textbf{Case $\varphi = \beta$:}  
    Assume $\beta = \sum_{\ell\in\mathcal X} c_\ell\cdot\sum_r \kappa(\ell,r)\le t$.  To show $\pi \models \varphi$ iff $\pi' \models \LTL(\varphi)$.
    \begin{itemize}
        \item ($\Rightarrow$) Suppose $\pi \models \varphi$ but $\pi' \not\models \LTL(\varphi)$. Then there exists some position $i$ in $\pi'$ with
    \[
    \sum_{\ell \in \mathcal{X}} c_\ell \cdot {\mathsf{cumul}_i}(\ell) > t[\mathcal{Y} \leftarrow \nu].
    \]
    Since $\mathsf{cumul}_i$ count exactly the total number of cumulative visits to each location till $\cfg_i$, the count is less than or equal to the total number of cumulative visits in $\pi$, thus 
    \[\sum_{\ell\in\mathcal X} c_\ell\cdot\sum_r \textsf{hsc}(\pi)(\ell,r) \geq \sum_{\ell \in \mathcal{X}} c_\ell \cdot \mathsf{cumul}_i(\ell) > t[\mathcal{Y} \leftarrow \nu]\] contradicting $\pi \models \varphi$.
        \item ($\Leftarrow$) Conversely, if $\pi' \models \LTL(\beta)$ then the number of process visits never exceeds $t[\mathcal Y\leftarrow \nu]$ at any prefix, so $\pi \models \beta$.
    \end{itemize}
\end{itemize}

Thus for every execution $\pi$ of $\RDTS^\BWA(\mathcal{T},\nu)$ and its corresponding execution $\pi'$ of $\RDTS^{\HRE}(\mathcal{T},\nu, \varphi)$:
$\pi \models \varphi$ iff $\pi' \models \LTL(\varphi).$
Quantifying over executions yields the desired equivalence at the model level.
\end{proof}

\subsection{Round-identity abstraction}
The system $\RDTS^\RIA(\mathcal{T},\nu, \varphi)$ abstracts the frontier-round identifier $\rmax$ from $\RDTS^\HRE(\mathcal{T},\nu, \varphi)$. For an $\HRE$ configuration 
\[\cfg_{\HRE}=\langle \cfg_\BWA,\, \local, \cumul\rangle\text{ with }
\cfg_\BWA = \langle \PState_{\BWA},\, \NState_{\BWA},\, r_{\max}\rangle,\] define the $\RIA$ projection as 
$\alpha_{\RIA}(\cfg_{\HRE}) \;=\; \langle \PState_{\BWA},\, \NState_{\BWA},\, \local, \cumul \rangle.$
Then the $\RIA$ semantics is the system \[\RDTS^\RIA(\mathcal{T},\nu, \varphi) = \langle S_\nu^{\RIA}, I_\nu^{\RIA}, Act_\nu^{\RIA}, Tr_\nu^{\RIA}\rangle, \] where 
$S_\nu^{\RIA}=\{\alpha_{\RIA}(\cfg) \mid \cfg \in S_\nu^\HR\}$, $I_\nu^{\RIA} = \{\alpha_{\RIA}(\cfg) \mid \cfg \in I_\nu^\HR\}$, and $Act_\nu^{\RIA} = Act_\nu^{\HRE}$.
The transition relation is the existential lift of $\HRE$ semantics:\\
$(\cfg_{\RIA},a,\cfg'_{\RIA})\in Tr_\nu^{\RIA}$ iff there exist $\cfg_{\HRE},\cfg'_{\HRE}\in S_\nu^{\HRE}$ with \[\alpha_{\RIA}(\cfg_{\HRE})=\cfg_{\RIA}, \alpha_{\RIA}(\cfg'_{\HRE})=\cfg'_{\RIA}, \text{ and }(\cfg_{\HRE},a,\cfg'_{\HRE})\in Tr_\nu^{\HRE}.\] %

\begin{lemma}[State-based bisimulation]\label{lem:RIA-bisim}  It holds that:
\begin{enumerate}
 \item  $\forall\ \cfg \in I_\nu^{\HRE}:$  $\alpha_{\RIA}(\cfg) \in I_\nu^{\RIA}$ and
  $ \forall\ \cfg_{\RIA} \in I_\nu^{\RIA}:$  $\exists\,\cfg \in I_\nu^{\HRE}$ with $\alpha_{\RIA}(\cfg)=\cfg_{\RIA}$.
\item For all $\cfg \in S_\nu^\HRE$ it holds:
\begin{itemize}
    \item[$\Rightarrow$] If \((\cfg,a,\cfg')\in Tr_\nu^{\HRE}\), then \((\alpha_{\RIA}(\cfg),a,\alpha_{\RIA}(\cfg')) \in Tr_\nu^{\RIA}\).
    \item[$\Leftarrow$] If $(\alpha_{\RIA}(\cfg),a,\cfg'_{\RIA})\in Tr_\nu^{\RIA}$, then $\exists\ \cfg'\in S_\nu^{\HRE},\ a' \in Act_\nu^\HRE$ such that\\ $(\cfg,a',\cfg')\in Tr_\nu^{\HRE}$ and $\alpha_{\RIA}(\cfg')=\cfg'_{\RIA}$.
\end{itemize}
\end{enumerate}
\end{lemma}
\begin{proof}
Part 1 is immediate from \(I_\nu^{\RIA}=\{\alpha_{\RIA}(\cfg)\mid \cfg\in I_\nu^{\HRE}\}\). 

Direction ($\Rightarrow$) follows from the definition of \(Tr_\nu^{\RIA}\) as the existential lift of \(Tr_\nu^{\HRE}\).
For direction (\(\Leftarrow\)), suppose \((\alpha_{\RIA}(\cfg),a,\cfg'_{\RIA})\in Tr_\nu^{\RIA}\). By definition of \(Tr_\nu^{\RIA}\) there exist \(\cfg_1,\cfg_2\in S_\nu^{\HRE}\) with 
\[(\cfg_1,a,\cfg_2)\in Tr_\nu^{\HRE}, \alpha_{\RIA}(\cfg_1)=\alpha_{\RIA}(\cfg),\text{ and }\alpha_{\RIA}(\cfg_2)=\cfg'_{\RIA}.\] Thus \(\cfg\) and \(\cfg_1\) agree on the $\PState_\BWA$, $\NState_\BWA$ and on history record counters $\local$ and $\cumul$; they may differ only on the frontier-round. Let \(\rmax^1\) be the frontier-round of \(\cfg_1\), \(\rmax\) be the frontier-round of \(\cfg\), and let \(\Delta =\rmax-\rmax^1\).
Suppose \(a=\Update(\rho,r)\). Set \(r' =r+\Delta\) and \(a'=\Update(\rho,r')\).
It follows from the semantics of $\RDTS^\HRE(\mathcal{T}, \nu,\varphi)$ that (i) \(a'\) is enabled at \(\cfg\), and (ii) executing \(a'\) at \(\cfg\) results in \(\cfg'\) such that \(\alpha_{\RIA}(\cfg')=\cfg'_{\RIA}\). This gives the required \(\cfg'\) and \(a'\). 
\end{proof}

\subsubsection{Proof of Proposition \ref{lem:RIA-completeness}.} The proposition follows from the Lem. \ref{lem:RIA-bisim} and the fact that $\alpha_\RIA$ preserves the state labels given by the history record counters.

\section{Technical proofs from Section~\ref{app:abstractions}}
\subsection{Proof of Lemma \ref{lm:RMA-refinement}} \label{apdx:ref-RMA}

\begin{proof}

Fix an $\RMA$ configuration $\cfg_{\RMA}=\langle \PState_{\RMA},\NState_{\RMA}\rangle\in S_\nu^{\RMA}$.  
Define $\beta_{\RMA}(\cfg_{\RMA})=\cfg=\langle \PState,\NState\rangle\in S_\nu$ by setting $\NState=\NState_{\RMA}$ and, for each $i\in\llbracket 1,N\rrbracket$, letting $\PState(i).\mathsf{loc}=\PState_{\RMA}(i).\mathsf{loc}$ and $\PState(i).\mathsf{rd}=\PState_{\RMA}(i).\mathsf{rd}$.

Write $\ell\leadsto_{0}\ell'$ if there is a (possibly empty) path of type-0 rules from $\ell$ to $\ell'$. Let
\[
X_i(\cfg)=\{\rho\in\mathsf{Rules}\mid \rho.\mathsf{type}=0\ \wedge\ \rho.\mathsf{to}\leadsto_{0}\PState(i).\mathsf{loc}\},
\]
and for a guard $\phi$, let $\mathsf{FV}_\mathcal{M}(\phi)\subseteq\mathcal M$ denote its set of message variables. Define
\[
M_i(\cfg)=\bigcup_{\rho\in X_i(\cfg)}\mathsf{FV}_{\mathcal M}(\rho.\mathsf{guard}).
\]

Finally, define the received-message multiset $\PState(i).\mathsf{rcvMsg}:\mathcal M\times\mathbb N\to\mathbb N$ by
\[
\PState(i).\mathsf{rcvMsg}(m,r)=
\begin{cases}
\NState(m,r), & r<\PState(i).\mathsf{rd},\\
\NState(m,r), & r=\PState(i).\mathsf{rd}\ \text{and } m\in M_i(\cfg),\\
0, & \text{otherwise.}
\end{cases}
\]

The $\beta_\RMA$ preserves the network component, respects process locations and rounds, and constructs received message component that satisfies condition~\eqref{equ:rec-bounded-bcast}.

$\beta_\RMA$ satisfies Part 1 of lemma, since initially every process is in initial location of round 0, and initial locations, and by Definition \ref{def:RPT} no type-0 rule can lead to an initial location and thus $M_i(\cfg) = \empty$ for all $i$.

\smallskip

For part 2, fix an abstract transition $(\cfg_{\RMA},a,\cfg'_{\RMA})\in Tr_\nu^{\RMA}$ with 
$a=\Update(i_0,\rho,r)$, and let 
$\beta_{\RMA}(\cfg_{\RMA})=\cfg$ and 
$\beta_{\RMA}(\cfg'_{\RMA})=\cfg'$, with $\cfg=\langle \PState,\NState\rangle$ and $\cfg' = \langle \PState',\NState'\rangle$.
By $\RMA$-enabledness we have, $\PState(i_0).(\mathsf{loc},\mathsf{rd})=(\rho.\mathsf{frm},r)$ and there exists 
$R:\mathcal M\times\mathbb N\to\mathbb N$ satisfying $0\le R\le \NState$ such that 
$R(\cdot, r)\models \rho.\mathsf{guard}[\mathsf P \leftarrow \nu]$.

\noindent \textbf{Step 1 (Prepare with receives).}
Let \(F=\mathsf{FV}_{\mathcal M}(\rho.\mathsf{guard})\). For each \(m\in F\), set
\(\Delta(m)=\max\{0,\,R(m,r)-\PState(i_0).\mathsf{rcvMsg}(m,r)\}\).
Execute exactly \(\Delta(m)\) transitions \(\Recv(i_0,\langle m,r\rangle)\) (for each \(m\in F\)),
obtaining \[\cfg \leadsto^{\Recv^\ast} \cfg^\ast=\langle \PState^\ast,\NState\rangle.\]
Feasibility holds because \(0\le\Delta(m)\le \NState(m,r)-\PState(i_0).\mathsf{rcvMsg}(m,r)\) by \(R\le\NState\) (see §\ref{receive-tr}).

Note that only \(\PState(i_0).\mathsf{rcvMsg}(\cdot,r)\) may change in $\cfg^*$ after series of process $i_0$ receives. Let \(S =F\cap M_{i_0}(\cfg)\).
By the definition of \(\beta_{\RMA}\):
\begin{itemize}
    \item For \(m\in S\), \(\PState(i_0).\mathsf{rcvMsg}(m,r)=\NState(m,r)\), hence \(\Delta(m)=0\) and therefore \(\PState^\ast(i_0).\mathsf{rcvMsg}(m,r)\ge R(m,r)\).
    \item For \(m\in F\setminus S\), \(\PState(i_0).\mathsf{rcvMsg}(m, r) = 0\), hence \(\Delta(m)=R(m,r)\) and therefore
\(\PState^\ast(i_0).\mathsf{rcvMsg}(m,r)=R(m,r)\).
\item 
For \(m\notin F\), and for all \(i\neq i_0\) or \(r'\neq r\), received message counts are unchanged.
\end{itemize}
Consequently, on the guard’s support \(F\), \(\PState^\ast(i_0).\mathsf{rcvMsg}(\cdot,r)\) agrees with \(R(\cdot,r)\) on \(F\setminus S\) and dominates it on \(S\).
Since \(\rho.\mathsf{guard}\) depends only on \(F\), and by Condition~$3$ in Def.~\ref{def:RPT}, the guard is
monotone in the variables from \(S\), replacing \(R(\cdot,r)\) by
\(\PState^\ast(i_0).\rcvMsg(\cdot,r)\) preserves satisfaction:
if \(R(\cdot,r)\models \rho.\mathsf{guard}[\mathsf P\leftarrow \nu]\), then also
\(\PState^\ast(i_0).\rcvMsg(\cdot,r)\models \rho.\mathsf{guard}[\mathsf P\leftarrow \nu]\).
Hence \(\Update(i_0,\rho,r)\) is enabled in \(\cfg^\ast\).

\noindent \textbf{Step 2 (Apply the update).}
From Step~1, \(a=\Update(i_0,\rho,r)\) is enabled in \(\cfg^\ast\).
Apply it to obtain
\[
\cfg^\ast \xrightarrow{\ \Update(i_0,\rho,r)\ } \cfg^x
  = \langle \PState^x,\NState^x\rangle.
\]
Receives do not change \(\NState\) or any \((\mathsf{loc},\mathsf{rd})\), and in both the concrete
and \(\RMA\) semantics the effect of \(\Update\) on these components coincides. Therefore, comparing
with the abstract step \((\cfg_{\RMA},a,\cfg'_{\RMA})\),
\[
\NState^x=\NState' \qquad\text{and}\qquad
\PState^x(i).(\mathsf{loc},\mathsf{rd})=\PState'(i).(\mathsf{loc},\mathsf{rd})
\ \ \text{for all } i.
\]

\noindent \textbf{Step 3 (Realign with receives).}
Only the receive message component may differ between \(\cfg^x\) and \(\cfg'\).
We reach \(\cfg'\) from \(\cfg^x\) by finitely many \(\Recv\) steps.

\noindent \emph{Past rounds.} 
For every process $i$ and every $r'<\PState'(i).\mathsf{rd}$, 
$\beta_{\RMA}(\cfg'_{\RMA})$ requires 
$\PState'(i).\mathsf{rcvMsg}(\cdot, r')=\NState'(\cdot, r')$. 
Since $\PState^x(i).\mathsf{rcvMsg}(\cdot, r')\le \NState'(\cdot, r')$ \\component-wise 
(Eq.~\ref{equ:rec-bounded-bcast}), we can deliver the outstanding messages 
at each $(m,r')$ until equality holds.

\noindent\emph{Frontier round.} 
Let $r_i=\PState'(i).\mathsf{rd}$ and let $M_i(\cfg')$ be as in the definition of $\beta_{\RMA}$; note that $M_i(\cfg')=M_i(\cfg^x)$. 
For each $i$, deliver messages at round $r_i$ exactly for those $m\in M_i(\cfg^x)$ 
until $\PState(i).\mathsf{rcvMsg}(m,r_i)=\NState'(m,r_i)$, 
and deliver none for $m\notin M_i(\cfg^x)$. 
This is feasible by Eq.~\ref{equ:rec-bounded-bcast}. 
Moreover:
\begin{itemize}
\item For $i\neq i_0$ we have $\PState^x(i)=\PState(i)$, hence $M_i(\cfg^x)=M_i(\cfg)$ and thus for all $m \notin M_i(\cfg^x)$ their receive count in round $r$ remain $0$.
\item For $i=i_0$: 
  \begin{itemize}
  \item If $\rho.\mathsf{type}=0$, then by construction of $M_{i_0}$ along type-0 paths, 
  $M_{i_0}(\cfg^x)\supseteq M_{i_0}(\cfg)\cup \mathsf{FV}_{\mathcal M}(\rho.\mathsf{guard})$, 
  so all messages delivered in Step~1 remain permitted and any additional ones required by $\beta_{\RMA}(\cfg'_{\RMA})$ can be added now.
  \item If $\rho.\mathsf{type}>0$, the frontier advances to $r+\rho.\mathsf{type}$, and since\\ 
  $\PState^x(i_0).\mathsf{rcvMsg}(\cdot, r+\rho.\mathsf{type})=0$, 
  we simply deliver up to $\NState'(\cdot, r+\rho.\mathsf{type})$ on the set $M_{i_0}(\cfg^x)$.
  \end{itemize}
\end{itemize}

\noindent\emph{Future rounds.} 
For all $i$ and $r''>\PState'(i).\mathsf{rd}$, $\beta_{\RMA}(\cfg'_{\RMA})$ requires \\
$\PState(i).\mathsf{rcvMsg}(\cdot, r'')=0$. 
This already holds in $\cfg^x$ and is preserved by not delivering any such messages.

After finitely many $\Recv$ transitions as above, the resulting concrete state 
$\cfg^\dagger=\langle \PState^{\dagger},\NState'\rangle$ satisfies 
$\PState^{\dagger}(i).(\mathsf{loc},\mathsf{rd})=\PState'(i).(\mathsf{loc},\mathsf{rd})$ 
and \\$\PState^{\dagger}(i).\mathsf{rcvMsg}=\PState'(i).\mathsf{rcvMsg}$ for all $i$; 
hence the state reached is exactly 
$\cfg^\dagger=\beta_{\RMA}(\cfg'_{\RMA})=\cfg'$.

We have thus constructed a concrete finite path
\[
\cfg\ \leadsto^{\Recv^*} 
\cfg^\ast\ \xrightarrow{\,a\,}\ \cfg^x\ \leadsto^{\Recv^*}\ \cfg',
\]
in which exactly one transition is $a$ and all others are $\Recv$, as required. 
\end{proof}

\subsection{Proof of Commutativity Claim (Appendix~\ref{app:SSR})}\label{apdx:soundness-SSR}
Before providing the proof we recall the $\PIA$ semantics. Recall in~$\PIA$ semantics, a configuration $\cfg=\langle \PState,\NState\rangle$ consists of
$\PState:\mathcal{L}\times\mathbb{N}\to\mathbb{N}$ (number of processes at location/round pairs $(\ell,r)$) and
$\NState:\mathcal{M}\times\mathbb{N}\to\mathbb{N}$ (number of messages of kind $m$ sent for round~$r$).
For an action $a=\Update(\rho,r)$, set
\[
s(a)=(\rho.\mathsf{from},r),\qquad
\tgt(a)=r+\rho.\mathsf{type},\qquad
t(a)=(\rho.\mathsf{to},\tgt(a)).
\]
Executing $a$ at $\cfg=\langle \PState,\NState\rangle$ produces $\cfg'=\langle \PState',\NState'\rangle$ with

\[
\PState'(\ell,r')=
\begin{cases}
\PState(\ell,r')- 1, &\text{if }(\ell,r')= s(a),\\
\PState(\ell,r') + 1 &\text{if }(\ell,r')= t(a), \\
\PState(\ell,r'), &\text{otherwise}.
\end{cases}
\]

\[
\NState'(m,r')=
\begin{cases}
\NState(m,r')+1,&\text{if }(m,r')=(\mathsf{Bcast}(\rho.\mathsf{to}),\tgt(a)),\\
\NState(m,r'),&\text{otherwise}.
\end{cases}
\]
Action $a$ is enabled at $\cfg$ if
\begin{itemize}
    \item (C1) $\PState(s(a))>0$, and
    \item (C2) there exists a witness $R:\mathcal{M}\times\mathbb{N}\to\mathbb{N}$ with $0\le R\le \NState$ (component-wise) such that $R(\cdot, r)\models \rho.\guard[\mathsf P\!\leftarrow\!\nu]$.
\end{itemize}

\noindent \textbf{We will use following two facts.}
Let $a=\Update(\rho,r)$:
\begin{enumerate}
\item (F1) Round locality of guards.
Enabledness of $a$ depends only on the round $r$ slice of $\NState$.
Indeed, if \(R\le\NState\) witnesses the guard, then so does $\widetilde R$ defined by $\widetilde R(\cdot,r)=R(\cdot,r)$ and $\widetilde R(\cdot,r'')=0$ for all $r''\neq r$.
\item (F2) Monotonicity in $\NState$.
If $a$ is enabled under $\NState$ and $\NState'\ge \NState$ component-wise, then $a$ remains enabled under $\NState'$ (reuse the same witness~$R$).
\end{enumerate}

\begin{proof}
Assume \((\cfg,a,\cfg^*),(\cfg^*,a',\cfg')\in Tr_\nu^\PIA\) with $\cfg= \langle \PState, \NState \rangle$, $\cfg^*= \langle \PState^*, \NState^* \rangle$, $\cfg'= \langle \PState', \NState' \rangle$, \(a=\Update(\rho,r)\) and \(a'=\Update(\rho',r')\), and assume \(\tgt(a)\neq r'\).

\medskip
\noindent\emph{Step 1: $a'$ is enabled at $\cfg$.}
Since $(\cfg^*,a',\cfg')\in Tr_\nu^\PIA$, enabledness of $a'$ at $\cfg^*$ gives
\[
\PState^*(s(a'))>0
\quad\text{and}\quad
\exists R^*\le \NState^*\text{ with }R^*(\cdot, r')\models \rho'.\guard[\mathsf P\leftarrow \nu].
\]

\noindent Enabledness condition C2 holds. Executing $a$ from $\cfg$ to $\cfg^*$ changes $\NState$ only at
$(\mathsf{Bcast}(\rho.\mathsf{to}),\,\tgt(a))$.
By the hypothesis $\tgt(a) \neq r'$, we have $\NState^*(\cdot, r')$ $=\NState(\cdot, r')$.
By (F1), $\widetilde R^*$ is a valid witness for the guard of $a'$ at $\cfg$.

\noindent Enabledness condition C1 holds.
If $s(a')\neq s(a)$, then $\PState^*(s(a'))=\PState(s(a'))$, hence $\PState(s(a'))>0$.
Else $s(a')=s(a)$, then $\PState^*(s(a))=\PState(s(a))-1>0$, so $\PState(s(a))\ge 2$ and in particular $\PState(s(a'))>0$.
Thus $a'$ is enabled at $\cfg$. Let $\cfg^\dagger=\effect_{\PIA}(\cfg,a')$. Then $(\cfg,a',\cfg^\dagger)\in Tr_\nu^\PIA$.

\medskip
\noindent\emph{Step 2: $a$ is enabled at $\cfg^\dagger$.}
Since $(\cfg,a,\cfg^*)\in Tr_\nu^\PIA$, enabledness of $a$ at $\cfg$ gives
\[
\PState(s(a))>0
\quad\text{and}\quad
\exists R\le \NState\text{ with }R(\cdot, r)\models \rho.\guard[\mathsf P\leftarrow \nu].
\]

\noindent Enabledness condition C2 holds. From $\cfg$ to $\cfg^\dagger$, executing $a'$ changes $\PState$ only at $s(a')$ (by $-1$) and $t(a')$ (by $+1$), and changes $\NState$ only by increasing the entry at $(\mathsf{Bcast}(\rho'.\mathsf{to}),\,\tgt(a'))$.
Hence $\NState^\dagger\ge \NState$ component-wise, so by (F2) any witness for the guard of $a$ at $\cfg$ remains valid at $\cfg^\dagger$.

\noindent Enabledness condition C1 holds.
If $s(a)\neq s(a')$, then
$\PState^\dagger(s(a))=$\\ $\PState(s(a))>0$.
Else $s(a)=s(a')$, then from Step~1 we had $\PState(s(a))\ge 2$, so
$\PState^\dagger(s(a))=\PState(s(a))-1\ge 1>0$.
Therefore $a$ is enabled at $\cfg^\dagger$, let $\effect_{\PIA}(\cfg^\dagger,a) = \widehat \cfg$ then $(\cfg^\dagger,a, \widehat \cfg)\in Tr_\nu^\PIA$.

\medskip
\noindent\emph{Step 3: $\widehat \cfg$ and $\cfg'$ are equal.}
Both sequences $aa'$ and $a'a$ modify exactly the same multiset of counters:
\begin{itemize}
\item on $\PState$: decrement at $s(a)$ and $s(a')$, increment at $t(a)$ and $t(a')$;
\item on $\NState$: increment at $(\mathsf{Bcast}(\rho.\mathsf{to}),\,\tgt(a))$ and at $(\mathsf{Bcast}(\rho'.\mathsf{to}),\,\tgt(a'))$.
\end{itemize}
Since these updates are additions of integers (decrement is addition of the count and -1) at fixed coordinates, which commute, we obtain
\[
\effect_{\PIA}(\effect_{\PIA}(\cfg,a),a') \;=\; \effect_{\PIA}(\effect_{\PIA}(\cfg,a'),a).
\]
By hypothesis, since the left-hand side is $\cfg'$ and the right hand side $\widehat \cfg$, hence $
\widehat \cfg=\cfg'$.

\medskip
Thus we have found $\cfg^\dagger$ with $(\cfg,a',\cfg^\dagger)$ and $(\cfg^\dagger,a,\cfg')$ in $Tr_\nu^\PIA$, as required. 
\end{proof}

\section{Explicit Encoding of Finite-Counter System}
\label{app:counter-system}

This appendix gives the explicit finite-counter-system construction obtained
after the reduction presented in Section~\ref{sec:abstractions}. The input is a process template
\[
\mathcal{T}
=
\langle
\mathsf{P},\mathsf{rc},\mathcal{L},\mathcal{I},
\mathcal{M},\mathsf{Bcast},\mathsf{Rules}
\rangle
\]
and an \(\HSCL\) formula \(\varphi\). For fixed-instance verification, the input
also includes a parameter valuation \(\nu:\mathsf{P}\to\mathbb{N}\).

The output is a finite-counter system 
\[
(V,\mathsf{INVAR},\mathsf{INIT},\mathsf{TRAN})
\]
together with the \(\LTL\) formula obtained from \(\varphi\), whose atomic
propositions are linear-arithmetic formulas over \(V\).
 Thus, the construction
produces an instance of \(\LTL\) model checking over a finite-counter system,
which can be encoded in symbolic model checkers such as nuXmv.

In the fixed-instance case, the invariant formula $\textsf{INVAR}$ additionally fixes the parameter variables
according to \(\nu\). This implies that only finitely many valuations satisfy \(\mathsf{INVAR}\). Hence, for fixed-instance verification, the generated finite-counter system is in fact finite-state.

\medskip
For a rule \(\rho\in\mathsf{Rules}\), write
$\rho=(\ell,\ell',k,g),$
where \(\ell=\rho.\mathsf{from}\), \(\ell'=\rho.\mathsf{to}\),
\(k=\rho.\mathsf{type}\), and \(g=\rho.\mathsf{guard}\). Let
\(
b=\max\{\rho.\mathsf{type}\mid \rho\in\mathsf{Rules}\}.
\)
Further, let \(p_1,\ldots,p_q\) be the parameter variables from \(\mathsf{P}\) that occur in \(g\), and let \(m_1,\ldots,m_s\) be the message variables from
\(\mathcal{M}\) that occur in \(g\). Let 
$\widehat g(p_1,\ldots,p_q,m_1,\ldots,m_s)$
be a quantifier-free Presburger formula equivalent to 
\[
\exists r_1,\ldots,r_s.\;
\left(
\bigwedge_{i=1}^{s} 0\leq r_i\leq m_i
\right)
\wedge
g[m_1\mapsto r_1,\ldots,m_s\mapsto r_s].
\]

Let \(X\subseteq\mathcal{L}\) be the set of locations mentioned in \(\varphi\) and let \(K\) be one plus the largest threshold appearing in cumulative atoms of \(\varphi\), see Def.~\ref{def:HSCL} (\(K=0\) if no such
atom exists).

Below, we describe each of the components of the obtained finite-counter system and the \LTL formula, including some brief explanations to help with understanding.

\paragraph*{Variables \(V\).}
The set \(V\) consists of the following non-negative integer variables:
\[
\mathsf{par}_p
\qquad
(p\in\mathsf{P}),
\]
\[
\mathsf{loc}_{\ell,d}
\qquad
(\ell\in\mathcal{L},\ d\in\{0,\ldots,b-1\}),
\]
\[
\mathsf{msg}_{m,d}
\qquad
(m\in\mathcal{M},\ d\in\{0,\ldots,b-1\}),
\]
\[
\mathsf{local}_{\ell},\mathsf{cumul}_{\ell}
\qquad
(\ell\in X).
\]
These are called, respectively, \emph{parameter variables, location variables, message variables, and history variables.} %

\paragraph*{Invariant formula \(\mathsf{INVAR}\).}
The formula \(\mathsf{INVAR}\) is the conjunction of the constraints below.

\noindent 1. The parameter variables satisfy the resilience condition:
\[
\mathsf{rc}\big[p\mapsto \mathsf{par}_p\big].
\]
For fixed-instance verification with parameter valuation
\(\nu:\mathsf{P}\to\mathbb{N}\), \(\mathsf{INVAR}\) additionally includes the
constraints
\[
\mathsf{par}_p=\nu(p)
\qquad
(p\in\mathsf{P}).
\]

\noindent 2. The location variables are bounded by the total number of processes:
\[
H \leq\sum_{\ell\in\mathcal L}\sum_{d=0}^{b-1}\mathsf{loc}_{\ell,d} \leq \mathsf{par}_n
\]
where \(H=n\) when failures are modeled explicitly by dedicated fault locations,
and \(H=n-t\) when failures are modeled implicitly by allowing processes to stop participating.

\noindent 3. The message variables are bounded by the maximum number of messages that
can be broadcast in one round:
\[
0\leq \mathsf{msg}_{m,d}\leq |\mathcal{L}|\cdot \mathsf{par}_n
\qquad
(m\in\mathcal{M},\ d\in\{0,\ldots,b-1\}).
\]

\noindent 4. The history variables are bounded as follows:
\[
(0\leq \mathsf{local}_{\ell}\leq \mathsf{par}_n)
\wedge (0\leq \mathsf{cumul}_{\ell}\leq K)
\qquad
(\ell\in X).
\]

\paragraph*{Initial formula \(\mathsf{INIT}\).}
The formula \(\mathsf{INIT}\) is the conjunction of the constraints below.

\noindent 1. The initial location variables satisfy
\[
\sum_{\ell\in\mathcal{I}}\mathsf{loc}_{\ell,0}
=
\mathsf{par}_n,
\]
\[
\mathsf{loc}_{\ell,0}=0
\qquad
(\ell\notin\mathcal{I}),\qquad \text{and}\]
\[\mathsf{loc}_{\ell,d}=0
\qquad
(\ell\in\mathcal{L},\ d\in\{1,\ldots,b-1\}).
\]

\noindent 2. The initial message variables satisfy
\[
\mathsf{msg}_{m,d}=0
\qquad
(m\in\mathcal{M},\ d\in\{0,\ldots,b-1\}).
\]

\noindent 3. The initial history variables satisfy
\[
\mathsf{local}_{\ell}
=
\begin{cases}
\mathsf{loc}_{\ell,0} & \text{if } \ell\in\mathcal{I},\\
0 & \text{otherwise,}
\end{cases}
\qquad
\mathsf{cumul}_{\ell}
=
\begin{cases}
\mathsf{loc}_{\ell,0} & \text{if } \ell\in\mathcal{I},\\
0 & \text{otherwise.}
\end{cases}
\]

\paragraph*{Transition formula \(\mathsf{TRAN}\).}
The transition formula is the disjunction of the local-transition formula and
the jump-transition formula:
\[
\mathsf{TRAN}
=
\mathsf{TRAN}_{\mathsf{loc}}
\vee
\mathsf{TRAN}_{\mathsf{jump}}.
\]

\medskip
\noindent {\bf Local-transition formula \(\mathsf{TRAN}_{\mathsf{loc}}\).}
Intuitively, the formula \(\mathsf{TRAN}_{\mathsf{loc}}\) encodes one process taking a
type-\(0\) rule in the current frontier round. For each type-\(0\) rule
\[
\rho=(\ell,\ell',0,g),
\]
we define a clause \(\mathsf{LocClause}_{\rho}\), and set
\[
\mathsf{TRAN}_{\mathsf{loc}}
=
\bigvee_{\rho=(\ell,\ell',0,g)\in\mathsf{Rules}}
\mathsf{LocClause}_{\rho}.
\]
The clause \(\mathsf{LocClause}_{\rho}\) is the conjunction of the
following constraints.

\noindent 1. The enabledness constrain:
\[
\mathsf{loc}_{\ell,0}>0
\;\wedge\;
\widehat g
\big[
p\mapsto \mathsf{par}_p,\;
m\mapsto \mathsf{msg}_{m,0}
\big].
\]

\noindent 2. The location variables are updated by moving one process from \(\ell\)
to \(\ell'\) at depth \(0\):
\[
\mathsf{loc}'_{\ell,0}
=
\mathsf{loc}_{\ell,0}-1,
\qquad
\mathsf{loc}'_{\ell',0}
=
\mathsf{loc}_{\ell',0}+1,
\]
and all other location variables are unchanged.

\noindent 3. If \(\mathsf{Bcast}(\ell')=m\in\mathcal M\), then the corresponding
message variable at depth \(0\) is incremented:
\[
\mathsf{msg}'_{m,0}
=
\mathsf{msg}_{m,0}+1,
\]
and all other message variables are unchanged. If
\(\mathsf{Bcast}(\ell')=\bot\), then all message variables are unchanged.

\noindent 4. If \(\ell'\in X\), then the history variables for \(\ell'\) is incremented as follows:
\[
\mathsf{local}'_{\ell'}
=
\mathsf{local}_{\ell'}+1, \qquad \mathsf{cumul}'_{\ell'}
=
\begin{cases}
\mathsf{cumul}_{\ell'}+1 & \text{if } \mathsf{cumul}_{\ell'}<K,\\
\mathsf{cumul}_{\ell'} & \text{otherwise.}
\end{cases}
\]
and all other history variables are unchanged. If \(\ell'\notin X\), then all
history variables are unchanged.

\noindent 5. Finally, all parameter variables are unchanged.

\medskip
\noindent {\bf Jump-transition formula \(\mathsf{TRAN}_{\mathsf{jump}}\).}
Intuitively, the formula \(\mathsf{TRAN}_{\mathsf{jump}}\) encodes a jump of the frontier
round. For each jump length
\[
h\in\{1,\ldots,b\},
\]
we define a clause \(\mathsf{JumpClause}_{h}\), and set
\[
\mathsf{TRAN}_{\mathsf{jump}}
=
\bigvee_{h=1}^{b}
\mathsf{JumpClause}_{h}.
\]

The clause \(\mathsf{JumpClause}_{h}\) uses auxiliary non-negetive integer variables
\[
J_{d,\rho}
\qquad
(d\in\{0,\ldots,b-1\},\ \rho\in\mathsf{Rules},\ \rho.\mathsf{type}>0).
\]
The variable \(J_{d,\rho}\) denotes the number of processes at depth \(d\) that
take the positive-type rule \(\rho\) during the jump. 

Formally, \(\mathsf{JumpClause}_{h}\) can be viewed as existentially
quantifying the auxiliary variables \(J_{d,\rho}\); the resulting formula can then be made linear-arithmetic by Presburger quantifier elimination.  In our encoding, we use the equivalent and more direct presentation: the variables \(J_{d,\rho}\) are included in \(V\) as auxiliary variables and are updated nondeterministically by jump transitions, so that their current values serve as witnesses for the jump.

The clause
\(\mathsf{JumpClause}_{h}\) is the conjunction of the following constraints.

\noindent 1. The auxiliary variables cannot move more
processes than are available. For every \(\lambda\in\mathcal L\) and $d\in\{0,\ldots,b-1\}$ writing
\[
J^{\mathsf{out}}_{\lambda,d}
=
\sum_{\substack{
\rho=(\lambda,\ell',k,g)\in\mathsf{Rules}\\
k>0
}}
J_{d,\rho},
\]
we require
\[
J^{\mathsf{out}}_{\lambda,d}
\leq
\mathsf{loc}_{\lambda,d}
\qquad
(\lambda\in\mathcal L,\ d\in\{0,\ldots,b-1\}).
\]

\noindent 2. A positive-type rule may be taken only if it lands exactly in the
new frontier round. Thus, for every
\(\rho=(\ell,\ell',k,g)\) with \(k>0\), and every
\(d\in\{0,\ldots,b-1\}\), we require
\[
J_{d,\rho}>0
\implies
d+h=k.
\]

\noindent 3. Whenever some processes take a rule, its guard must be enabled at
the corresponding depth. 
\[
J_{d,\rho}>0
\implies
\widehat g
\big[
p\mapsto \mathsf{par}_p,\;
m\mapsto \mathsf{msg}_{m,d}
\big],
\]

\noindent 4. The location variables are updated by shifting the old window by length
\(h\) and placing the processes that jump in the new frontier
round. For every \(\lambda\in\mathcal L\), let
\[
J^{\mathsf{in}}_{\lambda}
=
\sum_{\substack{
d\in\{0,\ldots,b-1\}\\
\rho=(\ell,\lambda,k,g)\in\mathsf{Rules}\\
k>0
}}
J_{d,\rho}.
\]
Then
\[
\mathsf{loc}'_{\lambda,0}
=
J^{\mathsf{in}}_{\lambda}
\qquad
(\lambda\in\mathcal L).
\]
For depths skipped by the jump, we set
\[
\mathsf{loc}'_{\lambda,d}=0
\qquad
(\lambda\in\mathcal L,\ 1\leq d<h).
\]
For the remaining depths, we shift the old variables and remove the processes
that jumped:
\[
\mathsf{loc}'_{\lambda,d}
=
\mathsf{loc}_{\lambda,d-h}
-
J^{\mathsf{out}}_{\lambda,d-h}
\qquad
(\lambda\in\mathcal L,\ h\leq d\leq b-1).
\]

\noindent 5. The message variables are updated analogously. The messages in the
new frontier round are exactly those broadcast by the processes that jump into
that round. Thus, for every \(m\in\mathcal M\),
\[
\mathsf{msg}'_{m,0}
=
\sum_{\substack{
\lambda \in\mathcal{L}\\
\mathsf{Bcast}(\lambda)=m
}}
J^{\mathsf{in}}_{\lambda}.
\]
For depths skipped by the jump, we set
\[
\mathsf{msg}'_{m,d}=0
\qquad
(m\in\mathcal M,\ 1\leq d<h).
\]
For the remaining depths, we shift the old message variables:
\[
\mathsf{msg}'_{m,d}
=
\mathsf{msg}_{m,d-h}
\qquad
(m\in\mathcal M,\ h\leq d\leq b-1).
\]

\noindent 6. The history variables are updated according to the processes that
enter locations mentioned in the formula. For every \(\lambda\in X\),
\[
\mathsf{local}'_{\lambda}
=
J^{\mathsf{in}}_{\lambda},\quad \text{and}\qquad \mathsf{cumul}'_{\lambda}
=
\begin{cases}
\mathsf{cumul}_{\lambda}+J^{\mathsf{in}}_{\lambda}
&
\text{if } \mathsf{cumul}_{\lambda}+J^{\mathsf{in}}_{\lambda}\leq K,\\[2mm]
K
&
\text{otherwise.}
\end{cases}
\]

\noindent 7. Finally, all parameter variables are unchanged.

\medskip
\noindent {\bf Remark.}
The nuXmv encoding additionally requires the transition relation to be total, i.e., to have no deadlock states. We ensure this manually by adding stuttering transitions wherever necessary.
\paragraph*{\LTL Specification.}
For completeness, we recall the \(\HSCL\)-to-\(\LTL\) translation presented 
Section~\ref{subsec:LTL}.  From \(\HSCL\) formula \(\varphi\), we
construct an \(\LTL\) formula \(\LTL(\varphi)\) recursively.  The atoms of the translated formula are linear-arithmetic formulas over the history variables of the finite-counter system, namely formulas of the form
\[
\sum_{\ell\in X} c_\ell \cdot \mathsf{local}_{\ell} \leq t
\qquad\text{or}\qquad
\sum_{\ell\in X} c_\ell \cdot \mathsf{cumul}_{\ell} \leq t,
\]
where \(c_\ell\in\mathbb{N}\) are constants,  \(t\in\LinTrm(\mathsf P)\), and $\mathsf{local}_{\ell},\mathsf{cumul}_{\ell}$ are history variables.

The translation is defined as follows. See Def.\ref{def:HSCL} for the syntax of \HSCL.  For a universal round-local atom
\[
\forall r.\ \alpha_r
\qquad\text{where}\qquad
\alpha_r =
\sum_{\ell\in X} c_\ell\cdot \kappa(\ell,r) \leq t,
\]
we define
\[
\LTL(\forall r.\ \alpha_r)
=
\mathbf{G}\!\left(
\sum_{\ell\in X} c_\ell\cdot \mathsf{local}_{\ell} \leq t
\right).
\]
For a cumulative atom
\[
\beta =
\sum_{\ell\in X} c_\ell\cdot \sum_r \kappa(\ell,r) \leq t,
\]
we define
\[
\LTL(\beta)
=
\mathbf{G}\!\left(
\sum_{\ell\in X} c_\ell\cdot \mathsf{cumul}_{\ell} \leq t
\right).
\]
Boolean connectives are translated compositionally:
\[
\LTL(\neg\psi)=\neg\LTL(\psi),
\qquad
\LTL(\psi\wedge\psi')=\LTL(\psi)\wedge\LTL(\psi').
\]

\end{document}